\documentclass[11pt]{article}
\usepackage{fullpage}
\usepackage{amssymb,amsmath}
\usepackage{epsfig}
\usepackage{palatino}
\usepackage{graphicx}
\usepackage{textcomp}
\usepackage{amsthm}
\usepackage{mdframed}
\usepackage{color}
\usepackage[colorlinks=true,linkcolor=blue,linktoc=all]{hyperref}
\newtheorem{theorem}{Theorem}
\newtheorem*{theorem*}{Theorem}
\newtheorem{proposition}{Proposition}

\newtheorem{definition}{Definition}
\newtheorem{algorithm}{Algorithm}

\newtheorem{lemma}{Lemma}
\newtheorem{conjecture}{Conjecture}

\newtheorem{fact}{Fact}

\def\calC{{\cal C}}

\def\F{\mathbb{F}}

\def\poly{{\rm poly}}
\def\log{{\rm log}}

\def\Span{{\rm Span}}

\def\tr{{\rm tr}}

\newcommand{\be}{\begin{eqnarray}}
\newcommand{\ee}{\end{eqnarray}}

\newcommand\ceil[1]{{\lceil #1 \rceil}}

\newcommand\ket[1]{{ |{#1} \rangle }}
\newcommand\bra[1]{{ \langle {#1} | }}
\newcommand\ketbra[1]{{\ket{#1}\bra{#1}}}

\def\P{{\sf{P}}}
\def\NLTS{{\sf{NLTS}}}

\def\qLTC{{\sf{qLTC}}}
\def\sLTC{{\sf{sLTC}}}
\def\LTC{{\sf{LTC}}}

\newcommand{\eps}{\varepsilon}
\renewcommand{\epsilon}{\varepsilon}

\title{Robust Quantum Entanglement at (nearly) Room Temperature}

\begin{document}

\author{Lior Eldar}

\maketitle

\abstract{
We formulate an average-case analog of the $\NLTS$ conjecture of Freedman and Hastings (QIC 2014)
by asking whether there exist topologically ordered systems with corresponding local Hamiltonians 
for which the thermal Gibbs state for constant temperature cannot even be approximated by shallow quantum circuits.

We then prove this conjecture for nearly optimal parameters:
we construct a quantum error correcting code whose corresponding (log) local Hamiltonian has the following property:
for nearly constant temperature (temperature decays as $1/\log^2\log(n)$) 
the thermal Gibbs state of that Hamiltonian cannot be approximated by any circuit
of depth less than $\log(n)$.
In fact, we show a stronger statement: that one can recover a bona-fide code-state from the Gibbs state by applying a shallow decoder.
In particular, it implies that appropriately chosen local Hamiltonians can give rise to ground-state long-range entanglement
which can survive without active error correction at temperatures which are nearly independent of the system size:
thereby improving exponentially over previously known bounds.

The proof introduces a new approach for placing lower bounds on the depth of quantum circuits that approximate quantum states - by demonstrating a shallow decoder for quantum error correcting codes.
This adds to the very few techniques available for showing such bounds (see e.g. Eldar and Harrow, FOCS 2017) and hence might be useful elsewhere.
Specifically,
the construction and proof combine quantum codes that arise from high-dimensional manifolds from the works of Hastings (ITCS 2017) and
Leverrier et al. (QIP 2019),
the local-decoding approach to quantum codes by Leverrier et al. (FOCS 2015)
and Fawzi et al. (STOC 2018)
and quantum locally-testable codes by Aharonov and Eldar (SICOMP 2013).

}

\section{General}

\item
In order to perform universal quantum computation, one should at the very least be able to store
quantum states for long periods of time.
While the Fault Tolerance theorem \cite{AB96} makes this possible using active error correction,
in parallel, and in part due to the limitations of the FT theorem (see e.g. \cite{AHH+02})
a huge research effort was devoted to finding quantum systems that can retain quantum information
passively - namely a self-correcting quantum memory.

Self-correcting quantum memories are often referred to as topologically-ordered systems (or TQO) which is a phase of matter
that exhibits long-range entanglement at $0$ temperature.
Since $0$ temperature states are essentially theoretical objects that one does not expect to encounter
in the lab, the race was on to find TQO systems whose long-range entanglement can survive
at very high temperatures - ideally at a constant temperature $T>0$ that is independent of the system size.

In recent years there has been progress in ruling out such robustness for low-dimensional
systems like the 2-D and 3-D Toric Code, but there has been an indication that perhaps in 4 dimensions and above,
robustness is more likely (see Section \ref{sec:beta} for a summary of these results).
Intriguingly, it seems that quantum mechanics does not fundamentally limit the ability
to store quantum states for long times, at least for high-dimensional systems.
Despite that
there remains today a large gap between our physical intuition and our ability to provide
formal proofs on the existence of robust systems. Hence
the problem of establishing the existence of robust TQO systems, even for high dimensions
is wide open.
In this work we focus on narrowing this gap.

\section{Topological Quantum Order (TQO)}


TQO is a phase of matter
(i.e. in addition to the traditional gas, liquid, and solid states)
defined as the zero eigenspace 
of a local Hamiltonian 
which is "robust" in the sense that
there can be no transition between orthogonal zero eigenstate without a phase transition
(See survey of TQO in \cite{Wen12}).
Formally one says that a system is $\eps$-TQO if any sufficiently local observable $O$ is unable
to discern orthogonal states of the groundspace - i.e. there exists some $z\neq 0$ such that
$$
\left\| P O P - z P \right\| \leq \eps
$$

In the language of quantum computing, TQO is mostly synonymous with quantum error correcting codes,
namely topological quantum codes - these are $\eps$-TQO systems with $\eps=0$.
Under error-correction terminology the TQO property of robustness of ground-state degeneracy is the quantum error-correcting
minimal distance: i.e. the system can retain its logical encoded state in the presence of sufficiently
small errors.

Thus, TQO systems have the promise that at zero temperature, their entanglement can passively sustain itself 
(i.e. without active error-correction) as 
a form of self-correcting quantum memory.  It is this stability that brought forth the 
immensely influential paradigm of the topological quantum
computer by Kitaev \cite{Kit03, FKLW02}, and even initiated large-scale
engineering efforts in trying to build such a set-up \cite{Nature16}.

Yet to every silver lining there is a cloud: while TQO is a zero-temperature phenomenon by definition, 
a physicist attempting to prepare a TQO state in a lab can only expect to encounter
a Gibbs state (i.e. a thermal state) of the
Hamiltonian governing the TQO for some low temperature.
Namely, instead of finding a ground-state of a TQO Hamiltonian $H$, she usually
encounters a state $e^{-\beta H}$ for some $\beta = 1 / \kappa T$, for some low temperature $T \to 0$, where $\kappa$
is the Boltzmann constant.

Hence, for TQO to serve as a self-correcting quantum memory a necessary property
for such a system is that it retains its long-range entanglement at some non-zero temperature $T>0$
that is independent of the size of the system.


A natural treatment of robustness of TQO systems can be made using quantum circuit lower bounds,
a form of analysis initially considered in \cite{Has11}.
Under TQO terminology, topologically ordered states {\it cannot} be generated from a tensor product
state using a shallow circuit, whereas a state is said to be "trivial" if it {\it can} be generated
from product states by shallow circuits - namely it is nearly equivalent to a product state
in its lack of quantum entanglement.	
With this terminology in mind we consider the following conjecture:

\begin{mdframed}
\begin{conjecture}\label{conj:1}

\textbf{Robust Circuit Depth for Topologically Ordered Systems}

\noindent
There exists a number $\beta>0$ and a family of 
topologically-ordered systems (local Hamiltonians) $\left\{ H_n \right\}_n$ on $n$ qubits
such that
for all $\gamma \geq \beta$ we have:
Any quantum circuit $U$ that approximates the thermal Gibbs state $e^{-\gamma H}$
to vanishing trace distance has depth $\Omega(\log(n))$.
\end{conjecture}
\end{mdframed}

In words: the conjecture posits the existence of a TQO system (say, a quantum error-correcting code) for which
one can show a circuit lower bound for the thermal state for all $T$ from $0$ (i.e. the ground-state) up to some constant temperature.
Such a system exhibits "robustness" in the sense that the circuit lower bound for approximating its Gibbs state 
does not collapse when temperature is increased.
The "global" nature of entanglement that a TQO system state might carry is captured in the requirement that the
minimal circuit depth is lower-bounded by a number that is logarithmic in the number of qubits, to allow potentially
the coupling of {\it any} pair of qubits in the system, using a locally-defined quantum circuit.

Similar variants of this conjecture have been studied in physics literature: for example, in \cite{Yoshida11}
Yoshida provides a negative solution to a similar conjecture for codes embeddable in $2$ or $3$ dimensional lattices.
On the other hand, in \cite{Has11, HWM14} the authors provided an indication to the affirmative of this conjecture by considering the $4$-dimensional 
Toric Code: for example,  Hastings \cite{Has11}  assumes the existence of certain error operators from which he derived a related property.
In \cite{AHH+08} the authors use an approximation of thermal systems called the weak Markovian limit,
and conclude that certain topological measurements are preserved at constant temperatures for exponentially
long time in the system size (exponential coherence times).
However, to the best of our knowledge there is no formal proof of conjecture \ref{conj:1}.

Here, we make progress towards affirming this conjecture formally by proving conjecture \ref{conj:1}
for nearly-optimal parameters:
\begin{mdframed}
\begin{theorem*}\label{thm:sketch}(sketch)

\noindent
\textbf{Robust TQO at nearly Room Temperature}

\noindent
There exists a $\log$-local family of 
quantum error-correcting codes $\left\{ {\cal C}_n \right\}_n$
and corresponding family of commuting local
Hamiltonians $\left\{ H_n \right\}_n$
such that for any 
$\beta_n = O(\log^2 \log(n))$ the following holds:
any quantum circuit $V$ that acts on $a \geq n$ qubits and 
approximates the thermal state at temperature 
at most $T_n = 1/(\kappa \beta_n)$
on a set of qubits $S$, $|S| = n$:
$$
\left\| \frac{1}{Z} e^{- \beta_n H_n} - \tr_{-S}( V \ketbra{0^{\otimes a}} V^{\dag}) \right\|_1 = o(1),
\quad
Z = {\rm tr}(e^{-\beta_n H_n})
$$
satisfies a circuit lower bound:
$$
d(V) = \Omega(\ln(n)).
$$
\end{theorem*}
\end{mdframed}
Our proof will actually show a stronger statement: namely that the thermal Gibbs state
of these codes, for sufficiently low, yet nearly constant temperature, can be decoded using
a shallow circuit to a bona-fide quantum code state.
This implies that the Gibbs state retains topological order (up to a shallow decoder) at very high temperatures:
if we initialize our system in some ground-state of the TQO, and allow it to thermalize, we can later
recover that code-state with little extra cost (see Section \ref{sec:imp} discussing the possible implementation error of such a set-up).
In particular, it implies that appropriately chosen local Hamiltonians can give rise to ground-state multi-partite entanglement
which can survive without active error correction at nearly-constant temperatures.
%
%

Thus the system above is, in a way, a third variant of the quantum error-correcting paradigm:
on one hand unlike active error correction it does not require active error-correction during storage, 
but on the other hand, and unlike passive error correction, it does require a single application of active error-correction before usage.

\subsection{Some Perspective}

\subsubsection{The Thermal Gibbs State}

This study explores quantum circuit lower bounds on arguably the most natural of physical states - namely the
thermal Gibbs state.
This state has been the subject of intense research in statistical physics, and in particular in quantum-mechanics. 
The thermal Gibbs state is a quantum state that can be thought of as the equilibrium state of a system formed by coupling a
ground-state of a physical system to a "heat bath" - meaning it is allowed to interact indefinitely with
an environment to which we have no access to.
\begin{definition}

\textbf{The Thermal Gibbs State}

\noindent
Let $H$ be a Hamiltonian and let $\left\{\ket{\psi_i}\right\}_i$ be an eigen-basis of $H$ with corresponding eigenvalues $E_i$.
For $T>0$ the thermal state of $H$, denoted by $e^{-\beta H}/Z$, $\beta = 1/\kappa T$ (where $\kappa$ is the Boltzmann constant) 
is a mixture of eigenstates of $H$ where
the probability of sampling $\ket{\psi_i}$ is proportional to $e^{- \beta E_i}$.
For $T = 0$ the Gibbs state can be any $\rho \in \ker(H)$.
\end{definition}
%

\subsubsection{The Regimes of "Inverse-Temperature" $\beta$}\label{sec:beta}

We consider here a summary of prior art: the temperature at which one can establish
an $\Omega(\log(n))$ (i.e. "global") circuit lower bound for the Gibbs state of a Hamiltonian:
\begin{center}
    \begin{tabular}{| l | l | l | l |}
    \hline
    Hamiltonian & Temperature & Result & Comments \\ \hline
    2 or 3-D systems & $O(1/\poly(n))$ & \cite{Yoshida11} & No-go theorem \\\hline
    Projective Code & $\Omega(1/\log(n))$ & This work & By definition, Without amplification \\ \hline
    Amplified Projective Code & $\Omega(1/\log\log(n))$ & This work & With amplification \\ \hline
     4-D Toric Code & $\Omega(1)$ & \cite{Has11} & Heuristic argument. \\ \hline

    \hline
    \end{tabular}
\end{center}

Our main theorem establishes the existence of log-local Hamiltonians for which the thermal state $e^{-\beta H}$
for $\beta = (\ln\ln(n))^2$ requires a logarithmic circuit depth.
Therefore it improves exponentially on previous work in terms of the provable highest temperature as a function of system size $n$ at which
circuit lower bounds can be maintained.

Notably, observe that the rate of errors experienced by
quantum states from this ensemble scales like $n/\poly\log(n)$ - i.e. a nearly linear fraction.
Such error rates result in error patterns whose weight is much larger than the minimal error-correcting distance of the quantum code, and hence it is
not immediately clear, at least from an information-theoretic perspective, whether these states - that formally cannot protect
quantum information - can be assigned circuit lower bounds.

One may try to artificially generate an example where such an error rate leads to non-trivial entanglement:
say by considering $n/\poly\log(n)$ tensor-product copies of a "good" quantum error-correcting code, where each copy is defined
on $\poly\log(n)$ qubits.
For such a system - a typical error would leave at least some good copy of the code intact, thereby leading to a circuit
lower-bound.
However, one can immediately see that such a bound can at best scale as $\log\log(n)$ in the number of qubits - 
because it is in no way a global phenomenon of the system - but rather a local "artifact" of the system.

Another simple example, leading to a much tighter bound, is the following: for $\beta = \log^3(n)$, and using a
$\qLTC$ with mild, say $1/\log^2(n)$ soundness (such as the code of Leverrier et al. \cite{LLZ18} that we use
here as the basis for our construction) the probability of sampling a bona-fide quantum code-state is
overwhelming.
One can then argue that this fact alone is sufficient to establish a circuit lower-bound.
Our theorem handles much lower values of $\beta$, namely $\log^2\log(n)$
where the typical error can have huge size, albeit still not a constant fraction of the total system.

\subsubsection{Is it Entangled ?}

Any quantum system satisfying Conjecture \ref{conj:1} has a highly entangled ground-space, because 
for $T=0$
a Gibbs state can be any (possibly pure) ground-state of a topologically ordered system.

That said, for $T>0$ 
the circuit lower bound is applied to mixed states: assigning a quantum circuit lower bound for the task of approximating a quantum mixed state
(as opposed to a pure state)
does not necessarily indicate the
existence of quantum correlations but rather the presence of long-range correlations, which may or may not be quantum.

However, as noted above even this somewhat weaker notion of a quantum system with a
highly-entangled ground-space that retains a quantum
circuit lower bound at very high temperature isn't known to exist (at least formally),
and is related to major open questions in quantum complexity theory.
See in this context the NLTS conjecture discussed in the section \ref{sec:NLTS}.

We stress again that the proof of Theorem \ref{thm:main} will establish a much stronger statement: namely that the thermal Gibbs state is in fact approximately
equivalent to a topologically-ordered state: all we need to do to recover it is to apply a very shallow decoder.
Therefore for all sufficiently low temperatures the thermal Gibbs state of the constructed Hamiltonian
is in fact highly entangled in a well-defined way.

\subsubsection{Quantum Circuit Lower Bounds}

Our current knowledge of unconditional quantum circuit lower bounds is very limited despite several works (see e.g. \cite{Nielsen06, EH17}), even compared
to the classical case.
For example, there are no known quantum analogs of structural lemmas on shallow classical circuits
like Hastad's switching lemma.
Recently there has been interest in demonstrating {\it classical} circuit lower bounds for quantum search problems \cite{Watts19},
but notably these are {\it classical} bounds that do not attempt to capture a quantum property of the circuit in question.

This work adds to the set of available tools for showing quantum circuit lower bounds - by combining 
quantum locally-testable codes, an analysis of the thermal state as a truncated Markov chain,
and a local decoder that relies on these two properties to decode a thermal state to a bona-fide
quantum code-state which can be assigned a circuit lower bound by information-theoretic arguments.
Previous works have either used quantum locally testable codes \cite{EH17} to argue direct circuit lower bounds,
or local decoders \cite{Has11, LLZ18} but as far as we know these strategies were never used
in conjunction.
We outline our strategy in more detail in Section \ref{sec:strategy}

\subsubsection{Energy versus error}

An important distinction that one needs to make early on is that having a quantum state with low-energy
{\it does not} necessarily imply it is generated by applying few errors to a ground-state.
This is only true if the Hamiltonian governing the quantum state is a so-called $\qLTC$  \cite{AE13}.
$\qLTC$'s are quantum analogs
of locally-testable codes (and see Definition \ref{def:qltc}).

Like their classical counterparts $\qLTC$'s are (local) Hamiltonians for which large errors necessarily result in a large number of violations
from a set of local check terms.

To give an example - consider Kitaev's 2-dimensional Toric Code \cite{Kit03} at very low-temperature, say $T = O(1/\sqrt{n})$.
At that temperature the probability of an error of weight $\sqrt{n}$, at least one which is composed of strings
of weight $\sqrt{n}$ is proportional to the probability of observing an error of constant energy, i.e.:
$$
e^{-\beta O(1)}
$$
i.e. comparable to the probability of a single error.
So, unless additional structure of the problem is used, for all we know the number of errors
could be $\Omega(n)$.
The reason for the above is that the Toric Code is known to have very poor soundness as a locally testable code: in fact one can have 
an error of size $\sqrt{n}$ with only two violations.

\subsubsection{The relation to $\NLTS$}\label{sec:NLTS}

\noindent
Conjecture \ref{conj:1} above is a mixed-state analog (albeit with a slightly more stringent requirement
on the circuit depth) of the $\NLTS$ conjecture due to Freedman and Hastings \cite{FH13} - which posits the existence of local Hamiltonians
for which {\it any} low-energy state can only be generated by circuits of diverging depth.

As far as we know neither conjecture is stronger than the other: 
On one hand, circuit lower bounds on pure-states cannot be used to deduce
circuit lower bounds on their convex mixtures (implying $\NLTS$ is at least as strong as conjecture \ref{conj:1}):
trivially,
one can consider a highly entangled eigenbasis of the entire Hilbert space.
The uniform mixture on such an eigenbasis is merely the completely mixed state which can be prepared by very shallow (classical) circuits.

On the other hand,
it is not clear how a quantum circuit lower bound on approximating a mixed-state implies
a similar bound for approximating a pure-state in its support
(implying that conjecture \ref{conj:1} is at least as strong as $\NLTS$).
Hence, to the best of our knowledge these two conjectures are very related, but formally incomparable.

Arguably, the only known strategy to establishing the $\NLTS$ conjecture,  outlined in \cite{EH17}, is to show a construction of
quantum locally-testable codes
($\qLTC$'s) with constant soundness and minimal quantum error correcting distance
scaling linearly in the number of qubits.
However, such a statement by itself requires the construction
of quantum LDPC codes with distance growing linearly in the number of qubits - a conjecture now open for nearly 40 years.
Thus our inability to make progress on ${\rm qLDPC}$ is a significant barrier to any progress on the $\NLTS$ conjecture.

In this work, we show that by considering a mixed-state (or average-case) analog of $\NLTS$
(while still placing a more stringent requirement on the circuit depth)
one can break away from this strategy using
the tools we already have today - namely $\qLTC$'s with $1/\poly\log$ soundness 
and code distance which is sub-linear in $n$, in this case $\sqrt{n}$,
and achieve a construction with nearly optimal parameters.
Nevertheless, it could be the case that the construction provided here is in fact $\NLTS$
- meaning there are no trivial states of the code Hamiltonian for sufficiently small constant temperature.

\subsubsection{Implementation Error}\label{sec:imp}

Above, we mentioned the ability of the proposed system to allow thermalization of an initialized state, and still be able to recover that state using a local decoder.
Arguably, one can argue against such a statement that one also needs to account for the implementation error of the Hamiltonian governing the TQO state,
as well as the decoding Hamiltonian.
There exist analogous claims against active fault-tolerance in the form of implementation error of the error-correcting unitary circuits.

However, we conjecture that the system we construct, 
insofar as the check terms of the Hamiltonians are concerned $\{ C_i\}$, 
is in fact robust against implementation error by virtue of local testability.
Recall that a locally testable code (see Definition \ref{def:qltc}) satisfies the following operator inequality:
\be
 \frac{1}{m}\sum_{i=1}^m C_i 
 \succeq   \frac{ s}{n}   D_{\calC}.
\ee
where $D_{\calC}$ is an operator that relates a state to its distance from the code-space.

On can check that if the LHS above suffers from an additive error quantified by a Hermitian error operator ${\cal E}$, 
$$
\frac{1}{m}\sum_{i=1}^m C_i 
+
{\cal E}
$$
then using standard results about stability of Hermitian operators under Hermitian perturbation,
the resulting code will still be,
for sufficiently weak error ${\cal E}$, a locally testable code albeit with slightly worse parameters, and for a smaller range of distances from the codespace:
it will not be able to faithfully test very small errors, but only large errors.

Still, this code will possess the key property that we use here to argue the main theorem: that for sufficiently low temperature (depending on $\|{\cal E}\|$) - the code reins in the error weight to small weights, and these errors are far apart to allow for local error correction.
Notably, the actual shallow decoder used in the argument will also suffer from implementation noise, undoubtedly, but as a theoretic argument about entanglement, it is only important to account for implementation error of the Hamiltonian, and not the decoding circuit.

Hence the system proposed has apparently two advantages: not only it is able to sustain long-range entanglement 
for high temperatures as established in Theorem \ref{thm:main}, one doesn't even need to implement it precisely to gain the first advantage.
We leave for future research to quantify precisely the degree to which this system is robust against implementation error.

\subsection{Some Open Questions}

We end this section with several questions for further research.
First, it is desirable to improve (reduce) the value of $\beta$ and improve (reduce) the locality of checks (currently they are log-local).

We note that a limiting factor to decreasing $\beta$ is the maximal size set for which one can show near-optimal expansion,
for $\qLTC$'s.  
In this work, using a constant-soundess $\qLTC$ we were only able to achieve such expansion for very small sets - sets of logarithmic size.
This is in part due to the fact that the underlying manifold of the code is the $n$-dimensional cube.
Perhaps then applying Hastings' construction to a different manifold - say
a cellulation of a hyperbolic manifold will yield a better soundness - expansion trade-off:
Given the successful use of hyperbolic manifolds in generating bipartite expanders (see e.g. \cite{L11})
it is conceivable that for such manifolds one can establish small-set expansion even for linear-size sets, while still maintaining a non-negligible soundness $1/\poly\log(n)$ similar to the original construction.

In order to reduce the locality of checks one could attempt to use a locality-reducing transformation proposed by Hastings \cite{Has16}.
However, such a transformation alas reduces the soundness of the $\qLTC$ hence defeating the purpose of the initial amplification,
which was crucial in gaining another exponential improvement (reduction) of $\beta$.
Hence, to reduce locality while retaining $\beta$ - a possible venue may be to 
break away from the space of commuting Hamiltonians and
approximate the log-local checks using perturbation gadgets,
which are non-commuting.

Second, we observe that our proof makes no particular use of the thermal Gibbs distribution except at a single point regarding the truncated Markov chain.  Hence we conjecture that our proof applies to a more general setting of distributions which are "sub-exponential" namely
$$
P(\tau) \leq e^{-\beta E(\tau)}.
$$
We raise as an open question what other classes of distributions can be assigned circuit lower-bounds using our techniques, possibly augmented by new ideas.

An interesting extension of this work is to extend it to actual quantum information - namely show the system can store an arbitrary quantum state
for long periods of time: notably here we have only showed that one can recover the uniform distribution on code-states, but it is not immediately
clear that it can preserve a single arbitrarily encoded code-state.
In addition, it would be insightful to understand the actual coherence time of such a system as a self-correcting quantum memory - we conjecture
that it is polynomial in $n$.

Finally, one could explore the possibility that the constructed code in fact satisfies the $\NLTS$ condition: namely that any low-energy state
is highly entangled.

\subsection{Overall Strategy}\label{sec:strategy}

In figure \ref{fig:flow} we outline the main steps of our argument.
To recall, the main goal of this study is to demonstrate that the thermal state $e^{-\beta H}/Z$
is hard to produce for sufficiently small $\beta$, and show the same for any ground-state of $H$ - mixed or pure.

Our overall strategy is to demonstrate a {\it shallow} quantum circuit
that allows to correct this thermal state
to some code-state of a quantum code with large minimal distance.
For a quantum code with large minimal error-correcting distance it is a folklore fact (made formal here)
that any quantum state in that codespace is hard to approximate
(the gray-shaded box in Figure \ref{fig:flow}),
thereby satisfying the hardness-of-approximation requirement for groundstates.
Furthermore, together with the existence of a shallow decoder, it implies a lower-bound on the circuit depth for $e^{-\beta H}$
as the lower-bound on a circuit generating a quantum code-state, minus
the depth of the decoder.
Thus, 
working in the diagram of Figure \ref{fig:flow} backwards we translate our overall theorem to demonstrating a shallow error-correcting circuit
from a thermal state to a code-state with polynomial distance.

\subsubsection{Translating Energy to Error}

The strategy outlined above requires us to demonstrate a shallow decoder for thermal states of sufficiently low temperatures.
Here we are faced with a severe obstacle: a thermal Gibbs state is defined in terms of energy, whereas the natural language for decoders is the language of "errors" (whether they are average-case or worst-case).
Hence we need a scheme to argue about the error distribution of the Gibbs state.

To our aid come quantum locally testable codes ($\qLTC$s) \cite{AE13} (and see Definition \ref{def:qltc}).
The main use of locally testable codes is to rein in the error weight of low-energy states.
We use this property in conjunction with the well-known Metropolis-Hastings algorithm (or MH)
on Hamiltonians corresponding to the check terms of $\qLTC$'s.
The MH algorithm is a standard tool in physics to simulate the thermal Gibbs state by a random
walk where transition probabilities between quantum states are dictated by their relative energies 
(see Section \ref{sec:Gibbs}).

Applying this tool to local Hamiltonians corresponding to $\qLTC$s 
we show that the thermal Gibbs state $e^{-\beta H}$,
for $H$ corresponding to a $\qLTC$,
can in fact be approximated by a so-called "truncated" MH process.
This is an MH process
where one truncates the evolution of errors when they reach some maximal error weight.
Hence, the Gibbs state is reformulated as a random walk of errors that is truncated
when the error weight is too large.
These arguments correspond to the top Vanilla-colored boxes in the diagram.

As a general note, as far as we know, no previous work using $\qLTC$'s made such a translation
from energy to error weight:
in \cite{EH17} the authors show that $\qLTC$'s with linear distance are $\NLTS$, but
since such codes are not known to exist, they
end up proving a somewhat weaker version called NLETS thus bypassing the energy-to-error
translation.
On the other hand, such a translation is probably the most natural way to proceed
w.r.t.  quantum codes:  we do not know how to treat "energies" on quantum states,
but if we can model the errors they experience we can leverage our vast knowledge
of quantum {\it error-correction} to handle them.
Hence, we believe that the use of the MH random process is of conceptual importance
and will be useful elsewhere,
since it allows for the first time, to bring the analysis from a point we want to argue
about ("energies") to a point where we have powerful analysis tools ("errors").

\subsubsection{Shallow Decoding from Local Expansion}

To recap the flow of arguments:
the arguments about the MH random process
(Vanilla-colored boxes in Figure \ref{fig:flow})
allow us to argue that the thermal Gibbs state $e^{-\beta H}$,
for a $\qLTC$ Hamiltonian $H$,
 can be simulated by sampling an error according to an MH random walk 
 that is controlled by the inverse temperature $\beta$ and the soundness of the $\qLTC$.
We would now like to leverage that property to demonstrate a shallow decoder for this state.

A key observation towards that end
(corresponding to the bottom Vanilla-colored box) is that while the MH random process is not an i.i.d. process, it does in fact conform to a somewhat weaker characterization called "locally-stochastic"
(or "locally-decaying")  \cite{LTZ15,FGL18,Got14},
which are a main source of inspiration of this work.
A noise is locally-stochastic if the probability of a large cluster of errors decays exponentially in its size (see Definition \ref{def:locallystochastic}).

Concretely, using the truncated MH random process we conclude (orange box) that for sufficiently large $\beta$ (low temperature) the errors experienced by the Gibbs state 
are locally-stochastic, and hence typically form only small clusters whose size is, say, at most $\log(n)$.

In effect, a stronger notion is true: we show that local-stochasticity of these random errors means that their clusters are also far away from each other - in the sense that even if we "blow up" each cluster by a factor of $1/\alpha$ 
(for small $\alpha>0$) they are still at most the size above.
This definition is called $\alpha$-subset and it too, is due to \cite{LTZ15, FGL18}.

In these works, error patterns that are locally-stochastic were shown to be amenable to correction by a local decoder, since intuitively, these errors can be "divided-and-conquered" locally.
In this context, the notion of $\alpha$-subsets
was used to handle the possibility that the decoder can introduce errors to qubits which weren't initially erred,
by arguing that even if such an event occurs it will not cause the initial clusters to aggregate together to form large,
undecodable clusters.

Our choice for a shallow "local-decoder" is to use
a straightforward quantum generalization of the Sipser-Spielman decoder (notably, a variant of this decoder was used in \cite{LTZ15}).
This decoder is desirable since it is able, under certain conditions to decode an error of weight $w$ in time $\log(w)$, and do so locally.
That would imply that for error patches of logarithmic size, the decoder would run in depth $O(\log\log(n))$ - i.e. a very shallow decoder.

However, such a decoder comes attached with a very stringent condition: it requires the Tanner graph of the code to be a very strong bi-partite expander.
That condition is too stringent for our purposes, since we also need the quantum code
to be locally-testable, and it is not known how to make even classical local-testability co-exist with the code's Tanner graph being a bi-partite expander.
\footnote{In fact, in \cite{SS96} the authors already acknowledge that their bi-partite code construction is unlikely to be a classical $\LTC$}.

In our study, we relax the stringent expansion condition, 
and require that only very small errors, i.e. those of logarithmic size which we've shown to be the typical error size for the Gibbs state - those errors are required to expand well ("small-set expansion"),
while requiring nothing for linear-weight errors, which is the regime of interest of the standard Sipser-Spielman decoder.
Hence, we
are able to use a code whose Tanner graph
is not a true expander. 
This while still being able to use the Sipser-Spielman approach to a fast parallel decoder by considering only small sets.
These arguments are outlined in the pink-shaded boxes in the middle of the diagram.

\subsection{The construction}

To recap again,
starting from the previous section, our goal is to find quantum a code,
whose thermal state can be corrected quickly and in parallel.
We've shown that if a code is $\qLTC$ then the Gibbs state can be essentially modeled as an error process that is locally-stochastic.
Locally-stochastic errors can be decoded quickly, if the underlying topology is a good expander - at least for the typical error size.
Hence,  
our interim goal is
to find a quantum code ${\cal C}$
that satisfies simultaneously three requirements:
\begin{enumerate}
\item
It has a minimal quantum error-correcting distance that is some polynomial in the number of qubits $n$, say $\sqrt{n}$ - 
to allow a circuit lower bound for proper code-states.
\item
It is {\it locally-testable} - 
to allow translation from energies to errors in the truncated-MH modeling of the Gibbs state.
\item
Expansion of the bi-partite Tanner graph corresponding
to the checks of the code, for errors of small weight (or "small-set expansion") - 
to allow for shallow decoding using the Sipser-Spielman algorithm.
\end{enumerate}

\subsubsection{The Choice of Quantum Code}

In \cite{Has17} Hastings found a way to make progress on the $\qLTC$ conjecture \cite{AE13}
by considering high-dimensional manifolds: he showed that by tessellating a high-dimensional sphere
using a regular grid (or some other topology for improved rate) the resulting quantum
code on $n$ qubits has soundness $1/\poly\log(n)$.
We make a crucial use of his approach here.

Recently, Leverrier et al. \cite{LLZ18} have proposed the projective code, which is an arguably simpler 
variant of this high-dimensional construction whereby a length-$3$ complex chain of $p$-faces
of the binary $N$-cube (modulo the all-ones vector), for $p = \Omega(N)$ is used to derive
a quantum code on $n$ qubits with distance scaling like $n^c$, for some constant $c>0$.
This code has improved soundness compared to the one in \cite{Has17}.
Our construction is based on the projective code on $n$ qubits, using $p$-faces of the $N$-dimensional
cube for $p = N/2$, where $N = \Theta(\log(n))$.

On one hand, by the minimal distance of the projective code one immediately
gains a circuit lower bound of $\Omega(\ln(n))$ on
the minimal depth circuit generating its ground-state
(corresponding to the gray-shaded block in Figure \ref{fig:flow}.
This satisfies the first requirement above.
It is also a $\qLTC$ with reasonable ($1/\log^2(n)$) soundness, (see navy-shaded block in Figure \ref{fig:flow})
thus satisfying the second requirement.

However, the last critical advantage that we gain by using this code, as opposed to say
the original high-dimensional manifold of Hastings, is not its improved soundness parameter
but rather the underlying structure of the high-dimensional cube:
namely its property of small-set expansion that exists in addition to its non-negligible soundness.
This property of small set expansion is the turnkey for allowing the application
of a shallow decoder to combat the typical errors of a thermal state with large $\beta$ parameter.

More specifically, we make crucial use of the structure of the $N$-cube to
establish the third requirement - namely,
 show that small error sets
expand significantly - i.e. have many unique incident constraints.
This emanates from the fact that small subsets of $p$-faces of the $n$-cube for $p = n/2$
have many adjacent $p+1$-faces and many adjacent $p-1$-faces.
As a technical aside, we note that contrary to Hastings' construction in \cite{Has17} our choice of the parameter $p$ for the $p$-faces
of the code is chosen to be precisely $N/2$, where $N$ is the dimension of the cube.
This, because it facilitates the proof of simultaneous local expansion
 for  the boundary map and  co-boundary maps.
 This corresponds to the green-shaded block in Figure \ref{fig:flow}.

\subsubsection{Degree Reduction to Reduce $\beta$ Exponentially}

The flow of arguments until this point results in lemma \ref{lem:maxconn} which roughly
states that a $\qLTC$ with soundness $s$ and qubit degree $D$ 
has the property that its thermal state 
for sufficiently large inverse temperature:
$$
\beta \geq \log(D) / s.
$$
has error patterns that form 
clusters of only logarithmic size.
Such errors admit a shallow-depth parallel decoding scheme resulting in a circuit lower bound
for approximating this thermal state.

Consider, for example the projective code of \cite{LLZ18}.  We have
$$
D = \log(n), \quad s = 1/\log^2(n)
$$
Using these parameters one would only be able to establish a circuit
lower bound for $\beta = \poly\log(n)$.
Hence, by the behavior of $\beta$ as a function of $D$ and $s$ one sees it is desirable to trade-off increased degree for improved soundness
so long as these two quantities are increased/decreased in a commensurate manner.

The seemingly natural way to do this is to define
a $\qLTC$ where where the set of checks corresponds to all possible subsets of $1/s$ checks,
and defining each check as satisfied only if all checks in its subset are satisfied.
Unfortunately, such a code has degree which scales like
$$
{n \choose 1/s} = 2^{O(\log^3(n))}
$$
which implies that $\beta$ needs to be at least  $\beta \geq \log^3(n)$ for the circuit lower-bound to hold.
However, using a standard probabilistic analysis we show that it is sufficient to choose a random
family of $\Omega( n \log^2(n))$ subsets - each set comprised of $1/s$ checks in order  to achieve a $\qLTC$
with constant soundness and query size $q/s$.
This corresponds to the second navy-shaded box in Figure \ref{fig:flow}.

This amplification procedure results in a somewhat peculiar situation that we'd like to point out:
the thermal Gibbs state $e^{-\beta H}$ is defined w.r.t. the Hamiltonian $H = H({\cal C}_{pa})$ where
${\cal C}_{pa}$ is the result of the amplification of the projective code
formed by choosing a sufficiently large random family of subsets of check terms of size $1/s$ each.
However, the decoding procedure, using the Sipser-Spielman algorithm uses the original checks
of ${\cal C}$ to locate and correct errors, and not the amplified ones: this is because
we do not establish local expansion for the amplified checks, only for the original checks.
Still, both sets of checks share the same code-space - namely the original projective code ${\cal C}$.
Hence, the set of checks used for {\it testing} are not the same as the ones used for {\it correcting errors}.

\begin{figure}
\center{
 \epsfxsize=5in
 \epsfbox{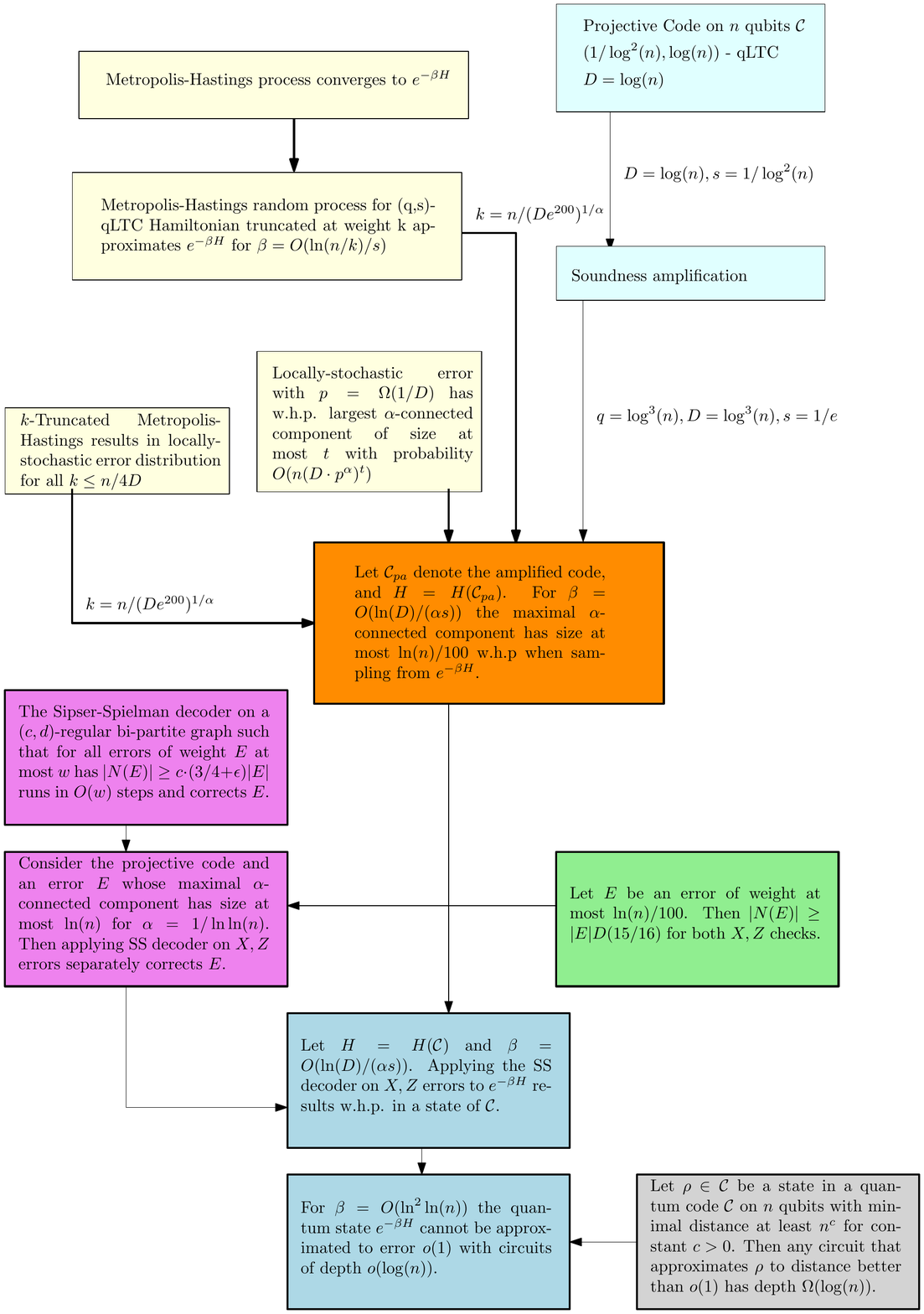}
 \caption{\footnotesize{
 Flow of the main arguments
 }\label{fig:flow}}}
\end{figure}

\section{Notation}

A quantum CSS code on $n$ qubits is a pair of codes ${\cal C} = ({\cal C}_x, {\cal C}_z)$,
where ${\cal C}_x, {\cal C}_z$ are subspaces of $\F_2^n$.
For a thermal state $\rho = (1/Z) e^{-\beta H}$, $\beta$ signifies the "inverse temperature" $\beta = 1/( \kappa T)$ where $\kappa$ is the Boltzmann constant,
and $Z = \tr(e^{-\beta H})$ is the partition function of this state.
For a finite discrete set $S$, $|S|$ denotes the cardinality of the set.
For ${\cal E}\in \{0,1\}^n$ the support of ${\cal E}$ , ${\rm supp}({\cal E})$ is the set of non-zero positions of ${\cal E}$.
$|{\cal E}|$ is the Hamming weight of ${\cal E}$.
For quantum density matrices $A,B$, the trace distance between $A,B$ is denoted by $\|A - B\|_1$,
and the quantum fidelity between these states is denoted by ${\cal F}(A,B)$.
A density matrix $\rho$ of rank $r$ is said to be a {\it uniform mixture} if it can be written as
$$
\rho = \frac{1}{r} \sum_{i=1}^r \ketbra{u_i}
$$
where $\{ \ket{u_i} \}$ are an orthonormal set of vectors.

We say that a quantum circuit $U$ on a set $T$ of $N$ qubits approximates a quantum state $\rho$ on a set $S \subseteq T$ of $n \leq N$ qubits, to error $\delta$
if 
$$
\left\| \tr_{T - S}\left(U \ketbra{0^{\otimes n}} U^{\dag}\right)  - \rho \right\|_1 \leq \delta
$$


In this work, we will consider random error models ${\cal E}$ supported on the $n$-th fold tensor product Pauli group ${\cal P}^n$,
where ${\cal P} = \left\{X,Z,Y,I \right\}$.
For an error ${\cal E}$ 
$$
{\cal E} = {\cal E}_1 \otimes \hdots \otimes {\cal E}_n, \quad {\cal E}_i \in {\cal P}
$$
we denote by $|{\cal E}|$ the Hamming weight of ${\cal E}$ - namely the number of terms ${\cal E}_i$ that are not
equal to $I$.
Often we will use $|{\cal E}|$ to denote the minimal weight of ${\cal E}$ modulo a stabilizer subgroup of ${\cal P}^n$.

For a stabilizer code ${\cal C}$ with local check terms $\{C_i\}_{i=1}^m$, $C_i \in {\cal P}^n$, 
the Hamiltonian $H = H({\cal C})$ is the local Hamiltonian
$\sum_{i=1}^m (I - C_i)/2$ - i.e. its ground-space is the intersection of the $1$-eigenspaces of all check terms $C_i$.

The $N$-cube is the binary cube in $N$ dimensions.  We will use capital $N$ to denote the dimension of the cube.
The projective cube results in a code of $n$ qubits.
When considering $n$ in the context of the projective cube we will use lower-case $n$ to denote the number of qubits, i.e. $n = 2^N$.

The letter $a$ will be used to denote an initial set of qubits $a \geq n$  that also include any ancillary qubits used to generate the state of $n$ qubits, 
i.e. to generate a mixed state on $n$ qubits one applies a unitary transformation on the state $\ketbra{0^{\otimes a}}$, traces out $a-n$ qubits.

Let $G = (V,E)$ be a graph.  For a set $S\subseteq V$ the set $\Gamma(S)\subseteq V$ is the set of all vertices that neighbor
$S$ in $G$.
The degree of a vertex $v\in V$ is the number of edges incident on that vertex.
The degree $D$ of a graph is the maximal degree of any vertex $v\in V$.
A graph is $D$-regular if the degrees of all vertices are equal.

We will use $a \propto b$ to signify that $a = c \cdot b$ for some $c$ that does not depend on $b$.

%

\section{Preliminaries I: Thermal Gibbs State of a Local Hamiltonian}

When considering the thermal Gibbs state for a local Hamiltonian $H = \sum_i H_i$, $\|H_i\|\leq 1$, care needs to be taken
as to how to scale the energy of the Hamiltonian.
On one hand, we would like the Gibbs state of a Hamiltonian $H$ to be invariant under scaling of $H$,
or perhaps rewriting $H$ as a sum of possibly lower-{\rm rank} projections.
On the other hand, we note that it is unreasonable to expect to have a family of local Hamiltonians
$\{H_n\}_n$ with
entanglement at room temperature (i.e. constant $\beta>0$), if the norm of $H_n$ doesn't grow with the number of qubits $n$.
In fact, in physics literature it has become a convention that when $m = O(n)$ the Gibbs state
of $H$ is merely $e^{-\beta H}$ - i.e. we allow $\|H \|$ to grow linearly in the number of qubits $n$.

Hence, we introduce the definition of energy density - which captures the average "energy"  invested
into a qubit in the system: 
\begin{definition}

\textbf{Energy density}

\noindent
A local Hamiltonian on $n$ qubits with $m$ local terms $H = \sum_{i=1}^m H_i$, $\|H_i \| \leq 1$
is said to have energy density $\lambda = m/n$.
\end{definition}

In this work, in the context of the Gibbs state, we will consider only Hamiltonians with {\it unit} energy density, i.e.
given $H = \sum_{i=1}^m H_i$, $\|H_i \| \leq 1$ we will consider instead
$$
\tilde H = \frac{n}{m} H
$$

The thermal Gibbs state is defined for a local Hamiltonian as follows:
\begin{definition}

\textbf{Gibbs state of a local Hamiltonian}

\noindent
Let $ H = \sum_{i\in [m]} H_i$, be a local Hamiltonian on $n$ qubits ${\cal H} = \mathbf{C}^{\otimes n}$, $m$ local terms, and energy density $\lambda = m/n$.
The Gibbs state of $H$ for finite $\beta>0$ is the following density matrix:
$$
\frac{1}{Z} e^{- \beta \tilde H}
$$
where $Z = \tr(e^{-\beta \tilde H})$ and $\tilde H = H / \lambda$.
For $\beta \to \infty$ the Gibbs state is any $\rho \in \ker(H)$.
\end{definition}

\section{Preliminaries II :Quantum  Error-Correcting Codes}

We require the basic definition of stabilizer codes and CSS codes
\begin{definition}\label{def:css}

\textbf{Quantum Stabilizer Code and Quantum CSS Code}

\noindent
A stabilizer group ${\cal G}\subseteq {\cal P}^n$ is an Abelian subgroup of ${\cal P}^n$.
The codespace ${\cal C}$ is then defined as the centralizer of ${\cal G}$, denoted by $C[{\cal G}]$, or equivalently - the mutual $1$-eigenspace of ${\cal G}$.
A CSS code ${\cal C} = ({\cal C}_x, {\cal C}_z)$ is a stabilizer code where the check
terms (i.e. generators of the group) are tensor-products of either only Pauli $X$ or only Pauli $Z$.
In particular regarding ${\cal C}_x, {\cal C}_z$ as $\F_2$ subspaces of $\F_2^n$ we have
${\cal C}_x \subseteq {\cal C}_z^{\perp}$ and vice versa.
\end{definition}

In this work, we will require some bounds on the minimal depth of a quantum
circuit to generate a quantum code state.
We recall a slight rephrasing of Prop. 45 in \cite{EH17} to mixed states:

\begin{lemma}\label{lem:depth}

\textbf{Robust circuit lower bound for CSS code-states}\cite{EH17}

\noindent
Let ${\cal C}$ be a quantum CSS code of non-zero rate $k>1$ on $n$ qubits with minimal distance $n^{\eps}$ for some $\eps>0$.
Let $\rho_{gs}$ be a mixture on a set of code-states of ${\cal C}$ and
let $V$ be a unitary circuit on $N\geq n, N = \poly(n)$ qubits that approximates $\rho_{gs}$:
$$
\rho = {\rm tr}_T\left(V \ketbra{0^{\otimes N}} V^{\dag}\right)
\quad
\|\rho - \rho_{gs}\|_{1} \leq n^{-2}, \quad \rho_{gs}\in {\cal C}
$$
Then the depth of $V$ is $\Omega(\ln(n))$.
\end{lemma}
We note that the lemma above doesn't actually use the full error-correction property of the state,
but rather a weaker property called distance-partition, whereby the distribution induced
by measuring $\rho_{gs}$ in some tensor-product basis is partitioned on far-away sets.

\subsection{Quantum Locally Testable Codes}

In \cite{AE13} Aharonov and the author defined quantum locally testable codes ($\qLTC$'s).
We state here a version due to Eldar and Harrow \cite{EH17}:
a quantum locally testable code can be defined by the
property that quantum states on $n$ qubits at distance $d$ to the codespace have energy
$\Omega(d/n)$.  
\begin{definition}\label{def:subspace-distance}
If $V$ is a subspace of $(\mathbf{C}^2)^{\otimes n}$ then define its
$t$-fattening to be 
\be V_t := \Span \{ (A_1 \otimes \cdots \otimes A_n)\ket\psi : \ket\psi \in V,
\# \{i : A_i \neq I\}\leq t\}.\ee
Let $\Pi_{V_t}$ project onto $V_t$.  Then define the distance operator
\be D_V := \sum_{t\geq 1} t (\Pi_{V_t} - \Pi_{V_{t-1}}).\ee
\end{definition}
This reflects the fact that for quantum states, Hamming distance
should be thought of as an observable, meaning a Hermitian operator
where a given state can be a superposition of eigenstates.

\begin{definition}\label{def:qltc}

\textbf{Quantum locally testable code}

\noindent
An $(q,s)$-quantum locally testable code ${\cal C}\subseteq (\mathbb{C}^2)^{\otimes n}$, 
is a quantum code with $q$-local projection $C_1,\ldots,C_m$
such that 
\be
 \frac{1}{m}\sum_{i=1}^m C_i 
 \succeq   \frac{ s}{n}   D_{\calC}.
\ee
$s$ is called the {\it soundness} parameter of the code.
\end{definition}
We note that the soundness parameter $s$ in this definition generalizes
the standard notion of soundness of a classical $\LTC$ as a special case, where all $C_i$'s are
diagonal in the computational basis.
In particular, if the quantum code is a stabilizer code, then the definition of quantum local testability can be further simplified to resemble classical local testability more closely:
\begin{definition}\label{def:sltc}

\textbf{Stabilizer Locally-Testable Codes ($\sLTC$)}

\noindent
An $\sLTC$ is a quantum stabilizer code that is $\qLTC$.
An equivalent group-theoretic of an $\sLTC$ is as follows:
${\cal C}$ is a stabilizer code generated by stabilizer group ${\cal G}$.
It is $(q,s)-\sLTC$ if there exists
a set $S$ of $q$-local words
in the stabilizer group $g_1,\hdots, g_t \in {\cal G}$ such that for  $P\in {\cal P}^n$
we have
$$
\P_{g \sim U[S]}\left(  [g, P] \neq 0 \right) \geq (|P|/n) \cdot s
$$
where $|P|$ is the Hamming weight of $P$ modulo the centralizer of $G$, $C[{\cal G}]$:
$$
|P| = \min_{z\in C[{\cal G}]} wt(P + z)
$$
where for $x\in {\cal P}^n$ $wt(x)$ counts the number of non-identity entries in $x$.
\end{definition}

See \cite{EH17} and \cite{AE13} for the derivation of this definition as a special case of Definition \ref{def:qltc}:
the operator ${\cal D}_{\cal C}$ penalizes a quantum state according to the "weighted" distance of that state from the codespace, whereas in definition \ref{def:sltc} the penalty is defined w.r.t. each Pauli error separately, and as a function of the standard Hamming weight of the error, modulo the code. 

Given a $(q,s)$-$\sLTC$ one can generate a $\sLTC$ with parameters
$(\ceil{q / s}, 1/e)$ by amplification as follows:
\begin{proposition}\label{prop:avg}

\textbf{Randomized Amplification}

\noindent
Given is a $(q,s)$  $\sLTC$ on $n$ qubits with $\poly(n)$ checks.
There exists a $\qLTC$  ${\cal C}_{amp}$ of $\poly(n)$ checks
with parameters 
$$
(\ceil{q /s}, 1/e)
$$
where each qubit is incident on at most
$D' = \ceil{q \log^2(n) / s}$ checks.
\end{proposition}

\begin{proof}

Assume w.l.o.g. that the set of checks $\{ C_i \}$ contain no repeats.
Recall that each check $C_i$ partitions the Hilbert space into states with eigenvalue $1$ (satisfying the check)
and eigenvalue $-1$ (violating the check).
Thus $(I + C_i)/2$ is a projection whose $1$-eigenspace is the satisfying space of $C_i$.
For a subset $S\subseteq [m]$ 
let $\tilde C_S$ be the projection operator whose $0$-eigenspace
is the intersection of $1$-eigenspaces of all checks $C_s$, $s\in S$:
$$
\tilde C_S = I -   \prod_{i\in S}  (I + C_i)/2 
$$
Hence, in particular, $\|\tilde C_S\| = 1$ for each $S$.
Fix some error ${\cal E}$, and set a conjugated erroneous state corresponding to ${\cal E}$ as follows:
$$
\tau = {\cal E} \cdot \rho_{gs} \cdot {\cal E} \quad \rho_{gs} \in {\cal C}
$$  
For $i\in [m]$ let $\chi_i$ denote the binary  variable that is $1$ if $\tau$ violates check $C_i$ - i.e. $\tr(C_i \tau) \neq 1$.
For a uniformly random set $S \sim U[ [m]^{1/s} ]$ let $\zeta$ denote
the binary random variable which is $1$ if 
$\tau$ violates $C_S$, i.e. $\tr(\tilde C_S \tau) \neq 0$.

Since ${\cal C}$ is a stabilizer $\qLTC$ then all checks $C_i$ commute so we have:
\be\label{eq:prod1}
\mathbf{E}[\zeta] = \P( \tr(\tilde C_S \tau) \neq 0) = \mathbf{E}_{S \sim U[[m]^{1/s}]}
 \left[ 1 - \prod_{i\in S} \left(1 -\chi_i \right)\right]
\ee
Let $\delta = \delta(\tau,{\cal C})$ denote the normalized distance between $\tau$ and ${\cal C}$, modulo ${\cal C}$:
$$
\delta(\tau,{\cal C}) = \Delta(\tau,{\cal C})/n
$$
Then by the $\qLTC$ condition of the original code ${\cal C}$ we have 
that a random check is violated by $\tau$ with probability
proportional to $\tau$'s distance from ${\cal C}$:
\be\label{eq:qltc1}
\forall \tau \quad
\mathbf{E}_{k \sim U[m]} \left[\chi_k\right] \geq s \cdot \delta(\tau,{\cal C})
\ee
Let $z$ denote a random variable corresponding to $\chi_i$ where $i\sim U[m]$.
Then by Equation \ref{eq:qltc1} $z$ is a Bernoulli random variable with bias
that is lower-bounded by $\tau$'s distance from ${\cal C}$:
$$
\forall \tau \quad
z \sim {\rm Bern}(p), \quad p \geq s \cdot \delta(\tau,{\cal C})
$$
then for a tuple of $|S|$ i.i.d. variables $z_1,\hdots, z_{|S|}$, each distributed as $z$, the random variable $\zeta$
is distributed by Equation \ref{eq:prod1} as:
$$
\zeta \sim 1 - \prod_{i\in [|S|]} \left(1 -z_k \right)
$$
and so
$$
\zeta \sim {\rm Bern}(q), \quad q \geq 1 - (1 - s \delta )^{|S|}
$$
setting $|S| = \ceil{1/s}$ we have:
$$
q \geq
1 - (1 - s \delta )^{|S| \delta / \delta}
\geq
1 - (1/e)^{\delta}
=
1 - e^{-\delta} 
\geq
1 - (1 - \delta/2)
=
\delta/2
$$
So we conclude:
$$
\zeta \sim {\rm Bern}(q), \quad q \geq \delta/2$$
Consider now a uniformly random family ${\cal F}$ of $m'$ checks where each check
is sampled randomly and independently from $[m]^{|S|}$.
By independence of choice of checks we have by the Chernoff-Hoeffding bound:
$$
\P_{\cal F}\left( \sum_{S\in {\cal F}} \zeta_S \leq \delta m'/e \right)
\leq
e^{-D(\delta/e || \delta/2)\cdot m'}
$$
using
$$
D(x || y) \geq \frac{(x-y)^2}{2y}
$$
implies
$$
D( \delta/e || \delta/2 ) \geq \frac{ (\delta \cdot 0.1)^2}{\delta} = \delta / 100
$$
On the other hand, for any $\delta$ the number of minimal-weight Pauli errors of fractional weight $\delta$ is at most
the number of (possibly not minimal) Pauli errors of weight at most $\delta$:
$$
4^{\delta n} \cdot {n \choose \delta n} \leq e^{ 2\delta n \log(n)}
$$
It follows by the union bound over all errors of fixed weight $\Delta n$ -  that if $m' \geq n \log^2(n)$ then 
the probability 
that some error of fractional weight $\delta$
has less than $\delta m'/e$
violations
is at most:
$$
e^{ 2\delta n \log(n)} \cdot e^{- \delta \cdot m' / 100} = 2^{-\Omega(n \log^2(n) \delta)}
$$
Since $\delta \geq 1/n$ it follows by the union bound over all values $\delta$ it follows there exists a family ${\cal F}_0$ such that
$$
\forall \tau 
\quad
\sum_{S\in {\cal F}} \zeta_S \geq \delta(\tau,{\cal C}) \cdot  m'/e
$$
Hence ${\cal F}_0$ is $\qLTC$ with soundness at least $1/e$, query size $\ceil{q/s}$ and at most $n\log^2(n)$ checks.
The degree of each qubit in ${\cal F}_0$
is at most 
$\ceil{q m' / (s n)}$.
\end{proof}

\section{Preliminaries III: Expansion of Small Errors on the Projective Hypercube}

The main observation of this section is that while the projective code is a $\qLTC$
with a mild soundness parameter $1/\log^2(n)$, the soundness parameter for {\it small}
errors is much better, and in fact for very small errors, their boundary (i.e. the Hamming weight
of their image) is very close to maximal.

We begin with a couple of standard definitions the first of which are the definitions
of the combinatorial upper and lower shadow of subsets of $r$ elements from a set of size $[n]$:
\begin{definition}

\textbf{Shadow}

\noindent
Let $[n]^r$ denote the set of all $r$-subsets of $[n]$,
and let ${\cal A} \subseteq [n]^r$.
The lower shadow of ${\cal A}$ is the set of all $r-1$ subsets
which are contained in at least one element of ${\cal A}$:
$$
\partial^{-}{\cal A} = \{A - \{i\} : A\in {\cal A}, i\in A\}
$$
and the upper shadow of ${\cal A}$ is the set of all $r+1$ subsets
that contain at least one element of ${\cal A}$:
$$
\partial^{+}{\cal A} = \{A + \{i\} : A\in {\cal A}, i\notin A\}
$$
\end{definition}

\noindent
We define $p$-faces as follows:
\begin{definition}

\textbf{$p$-face, set of $p$-faces, subspaces of $p$-faces}

\noindent
For integer $n\geq 1$ a $p$-face is a word in $\{0,1,*\}^n$ that contains exactly $p$ positions with $*$.
We denote by ${\cal K}_p^N$ as the set of $p$-faces of the $n$-th cube.
Let $C_p^N$ denote the space spanned
by ${\cal K}_p^N$ with coefficients from $\F_2$.
\end{definition}

\noindent
One can think about a $p$-face as a subset of $\{0,1\}^n$ of all points that are equal to the $p$-face in its non-$*$ positions.
Under this notation one can naturally define upper and lower shadow of $p$-faces as follows:

\begin{definition}

\textbf{Shadow of $p$-faces of ${\cal K}_p^N$}

\noindent
The lower-shadow $\partial^-$ of a $p$-face $f$ is the set of all $p-1$ faces derived
by replacing any $*$ entry with either $0$ or $1$.
The upper-shadow $\partial^+$ of a $p$-face $f$ is the set of all $p+1$ faces that can be derived
by replacing any non-$*$ entry of $p$ with $*$.
\end{definition}

To connect the definitions above, note that 
the $\F_2$-boundary operator $\partial_{p+1}$ associated with the $\F_2$-complex chain $\{ C_{p}^n\}_p$ maps each $p+1$-face $f$
to a summation over the set of $p$-faces $\partial^- f$ with coefficient $1$ in $\F_2$,
whereas the co-boundary map $\partial_p^T$ sends each $p-1$ face $f$ to a summation
over the set of $p$-faces $\partial^+ f$ with coefficient $1$.

Importantly, in this work, we will focus on the $p$-faces of the projective cube as the combinatorial set ${\cal K}_p^N$,
and not on the corresponding $\F_2$-space $C_p^N$.
This is because we are interested in establishing a combinatorial expansion property of the boundary maps $\partial^+, \partial^-$,
to be later used in conjunction with the Sipser-Spielman decoder.

However, we will use, in a black box fashion, the properties of these maps, as maps over an $\F_2$ complex chain that appeared 
in $\cite{LLZ18}$: these properties are namely the soundness and
minimal distance of a quantum code derived by the pair $(\partial^+, \partial^-)$.

For completeness, we mention some of the main results pertaining to the expansion of $\partial^+, \partial^-$
as combinatorial sets.
A central result in extremal combinatorics is the Kruskal Katona theorem.
That theorem asserts, in a version due to Lovasz, that when considering a family ${\cal A}$ of subsets of $[n]$ of size $r$, 
the size of the lower shadow of ${\cal A}$ behaves essentially like choosing subsets of elements of ${\cal A}$ of size
$r-1$ without repetition.
\begin{lemma}\label{lem:KK}

\textbf{Kruskal-Katona theorem}

\noindent
Let ${\cal A} \subseteq [n]^r = \{1,2,\hdots, n\}^{(r)}$, and $x\geq 0$ such that $|{\cal A}| = {x \choose i}$.
Then the lower shadow of $A$ satisfies
$$
|\partial^- {\cal A}| \geq {x \choose i-1}.
$$
\end{lemma}

\noindent
Subsequent theorem by Bollobas extends this to upper-shadows:

\begin{lemma}

\textbf{Bollobas' extension}

\noindent
Let ${\cal A} \subseteq [n]^r = \{1,2,\hdots, n\}^{(r)}$, and $x\geq 0$ such that $|{\cal A}| = {x \choose i}$.
Then the upper shadow of $A$ satisfies
$$
|\partial^+ {\cal A}| \geq {x \choose i +1}.
$$
\end{lemma}

In our case, however we would like to treat $p$-faces of the $n$-hypercube.
While this resembles the case of subsets of $[n]^r$ there is a major difference - since
now any $*$-entry replaced, can assume a value either $0$ or $1$, and the isoperimetric inequality
needs to account for this larger set.
Bollobas and Radcliffe provide an isoperimetric inequality for the regular grid.
\begin{lemma}

\textbf{Isoperimetric inequalities for $r$-faces}[Thm. 10, Bollobas and Radcliffe]

\noindent
Let $n>0$ and $E \subseteq {\cal K}_p^n$ be a set of $p+1$-faces of the $n$-th hypercube such that $|\partial^- E| = 2^{y-p} {y \choose p}$
then
$$
|E| \leq 2^{y-p-1} {y \choose p+1}
$$
\end{lemma}

The bounds above are useful especially when the set of faces is exponentially large
in the dimension of the embedding space.
For our purposes though, 
we are interested in set of $p$-faces that are polynomial in that dimension.
In such a scenario,
a much simpler bound is available as follows:

\begin{lemma}\label{lem:shadowsize}
Let ${\cal A} \subseteq {\cal K}_{p-1}^n$ be a set of $(p-1)$-faces for $p = n/2$, $|{\cal A}| \leq n/32$.
Then
$$
|\partial^+ {\cal A}| \geq |{\cal A}| \cdot (n/2+1) \cdot (15/16)
$$
Let ${\cal A} \subseteq {\cal K}_{p+1}^n$ be a set of $p+1$-faces for $p = n/2$, $|{\cal A}| \leq n/8$.
Then
$$
|\partial^- {\cal A}| \geq 2 \cdot |{\cal A}| \cdot (n/2+1) \cdot (15/16)
$$
\end{lemma}

\begin{proof}

Every pair of $p-1$-faces $f_1,f_2$ share at most a single $p$-face,
under the map $\partial^+$.
It follows that the number of $p$-faces in $\partial^+ {\cal A}$
that have more than a single $p-1$ face mapped to them is at most
$$
{ |{\cal A}| \choose 2}  \leq |{\cal A}|^2
$$
Let $D_+$ denote the degree of each $p-1$ face under the map $\partial^+$.
Then
$$
D_+ = n/2 + 1
$$
Since $|{\cal A}| \leq n/32$ then
$|{\cal A}| \leq D_+ / 16$ so we have
$$
|{\cal A}|^2 \leq  |{\cal A}| \cdot D_+ / 16
$$
and so the number of unique $p$-faces neighboring ${\cal A}$ under $\partial^+$
is at least
$$
|{\cal A}| \cdot D_+  - |{\cal A}| \cdot  D_+ / 16  = |{\cal A}| \cdot D_+ \cdot (15/16).
$$
hence
$$
|\partial^+ {\cal A}| \geq |{\cal A}| \cdot D_+ \cdot (15/16).
$$
Similarly every pair of $p+1$-faces $f_1,f_2$ share at most one $p$-face
under the map $\partial^-$.
Let $D_-$ denote the degree of each $p+1$ face under the map $\partial^-$.
Then
$$
D_- = 2 (n/2 + 1)
$$
Since $|{\cal A}| \leq n / 32$ then $|{\cal A}| \leq D_- / 16$
so we similarly have
$$
|\partial^- {\cal A}|  \geq |{\cal A}| \cdot D_- \cdot (15/16).
$$

\end{proof}

\subsection{The Projective Code}

\begin{definition}

\textbf{The Projective Cube}

\noindent
Let ${\cal K}_p^N$ denote the set of $p$-faces of the $N$-th cube.
The projective cube, denoted by $\tilde {\cal K}_p^N$ is formed by identifying
$$
x \sim \bar{x}  \quad \mbox{iff} \quad x = \bar{x} + \mathbf{1}
$$
Let $\tilde C_p^N$ denote the space spanned by $\tilde {\cal K}_p^n$ with coefficients in $\F_2$.
\end{definition}

In this study, we will use build upon the projective code defined by Leverrier et al. \cite{LLZ18}:
\begin{definition}\label{def:proj}

\textbf{Projective code}

\noindent
Extend the operators $\partial^+, \partial^-$ from ${\cal K}_p^N$ to $\tilde {\cal K}_p^N$ and
consider the complex chain formed by the $\F_2$ span of $\tilde {\cal K}_p^N$, namely the spaces $\{\tilde C_p^N\}_p$:
$$
\tilde C_{p+1}^N \to^{\partial_{p+1}} \tilde C_p^N \to^{\partial_{p}} \tilde C_{p-1}^N
$$
the quantum CSS code (see Definition \ref{def:css}) defined by ${\cal C}_x = \ker(\partial_{p}), {\cal C}_z = ({\rm Im}\partial_p)^{\perp}$ is called the  $(N,p)$-projective code
and denoted by $ {\cal C}_{N,p} = ( {\cal C}_x,  {\cal C}_z)$.
\end{definition}

\noindent
In \cite{LLZ18} the authors show the following:
\begin{lemma}\label{lem:LLZ18}

\textbf{Properties of the projective code}

\noindent
For every sufficiently large $N$ there exists $n = 2^{\Omega(N)}$ such that
the $(N,p)$-projective code $ {\cal C}_{N,p}$ for $p = N/2$ has parameters $[[n,1,n^c]]$, for some constant $c>0$.
It has soundness $1/\log^2(n)$ and each qubit is incident on at most $D = 2\log(n)$ checks.
\end{lemma}

We conclude this section 
by reducing the isoperimetric inequality for the projective cube to the isoperimetric
inequality for the $N$-cube.
\begin{lemma}\label{lem:projexp}

\noindent
\textbf{Isoperimetric inequalities for the projective hypercube}

\noindent
Let ${\cal C} = ({\cal C}_x, {\cal C}_z)$ denote the $(N,p)$-projective code with $p = N/2$.
Let ${\cal E}$ be a subset
of errors of weight at most $N/64$.
Then the number of checks ${\cal C}_x$ incident on ${\cal E}$ is at least
$$
|{\cal E}| \cdot (N/2) \cdot (15/16)
$$
and the number of ${\cal C}_z$ checks incident is at least
$$
|{\cal E}| \cdot (N) \cdot (15/16)
$$
\end{lemma}

\paragraph{Some context:}
To provide some context,
we note that at first sight it is unclear why considering such small weight ($N/64$) may provide a non-trivial result:
after all, for the regime of temperatures we are considering the typical error has nearly linear weight - i.e.
$n/\poly\log(n)$, and since $n = 2^N$ the weight considered above is merely $\poly\log(n)$.
The reason is that as we later show in the proof, the typical error of the Gibbs state is not arbitrary,
but can be further characterized as being formed on very small clusters - clusters of logarithmic size
(see Lemma \ref{lem:maxconn}).
We would like the check terms of the $p$-th projective code ${\cal C}_{N,p}$ to be such that {\it any}
error of logarithmic size expands very well in the Tanner graph of the code.
The isoperimetric inequality provided here on this very restricted error model will allows us to argue that we can use
a Sipser-Sipelman type decoder to correct all errors of the thermal state with high probability.

\begin{proof}
For each $x\in {\cal K}_p^N$ let $\tilde x \in \tilde {\cal K}_p^N$ denote the representative class of $x$,
namely 
$$
\forall x,y \in {\cal K}_p^N  \quad x = y + \mathbf{1} \quad \iff \quad \tilde x = \tilde y
$$
where we assume the convention that $* + 1 = * + 0 = *$ as in \cite{LLZ18}.

For $\tilde x \subseteq \tilde {\cal K}_p^N$ 
define a corresponding set $x(\tilde x) \subseteq {\cal K}_p^n$ in the cube 
$$
x = x(\tilde x) = \tilde x \cup (\tilde x + \mathbf{1}) \subseteq {\cal K}_p^N
$$
and define the maps $\partial^+, \partial^-$ by extension as
\be\label{eq:extend}
\partial^+ \tilde x = \widetilde {\partial^+ x}
\quad
\partial^- \tilde x = \widetilde {\partial^- x}
\ee
where the equivalence function is applied on each element of $\partial^+ x$ in RHS.

Consider some subset $\tilde E \subseteq \tilde {\cal K}_{p+1}^n$, 
$|\tilde E| \leq N/64$
and consider the number of incident $p$ faces
in $\tilde {\cal K}_p^N$
via the lower shadow $\partial^-$.
$\tilde E$ may be formed by identifying pairs of $p+1$ faces from a subset $E = E(\tilde E)\subseteq {\cal K}_{p+1}^N$, 
such that each $x\in E$ has a unique $y\in E$, $y = x + \mathbf{1}$ as follows:
$$
E(\tilde E) = \tilde E \cup (\tilde E + \mathbf{1})
$$
On one hand, we have by definition of $E$:
$$
|\widetilde E| = \frac{1}{2} |E|
$$
and on the other hand we observe that for sets $E$ of antipodal words
the boundary $\partial^- E$ is also comprised of antipodal words hence:
$$
|\widetilde{\partial^-  E}| = \frac{1}{2} |\partial^- E|
$$
so together with Equation \ref{eq:Econd} and the definition of extension in Equation \ref{eq:extend}
we get:
$$
\frac{|{\partial^-  \tilde E}|}{|\tilde E|} 
=
\frac{|\widetilde{\partial^-  E}|}{|\tilde E|} = \frac{|\partial^- E|}{|E|} 
$$
Finally, since $|\tilde E|  \leq N/64$ then
\be\label{eq:Econd}
|E(\tilde E)| = 2 |\tilde E| \leq N/32,
\ee
and
so by Lemma \ref{lem:shadowsize} we conclude:
$$
\frac{|{\partial^-  E}|}{ |E|} 
\geq N \cdot (15/16)
$$
and so
$$
\frac{|{\partial^-  \tilde E}|}{|\tilde E|} 
\geq N \cdot (15/16)
$$
The same argument holds for the upper shadow $\partial^+$.
\end{proof}

\section{Behavior of Errors in the Gibbs State of $\qLTC$s}\label{sec:Gibbs}

\subsection{The Thermal Gibbs by the Metropolis-Hastings Algorithm}

As mentioned in the introduction, a recurring barrier in the emergent field of robust quantum entanglement
is to establish a connection between the energy of a state, w.r.t. some local Hamiltonian,
and the "error" experienced by that state.

The main observation in this section is that specifically for $\qLTC$'s the Gibbs state
can be formulated as a random error process (and specifically, a discrete finite Markov process)
where the errors occur independently at each step, 
with an error rate that is comparable to the energy parameter of the state.
This will then allow us to conclude that for sufficiently small energy of the Gibbs state
the resulting errors can only form very small clusters.
We begin with the following standard definition:
\begin{definition}

\textbf{The Metropolis-Hastings Random Process Stabilizer Hamiltonians}

\noindent
Let ${\cal G}$ be a stabilizer group with a corresponding Hamiltonian $H = H({\cal G})$ on $n$ qubits $H = \sum_{i=1}^m H_i$ 
with $m$ local terms, and $\lambda(H) = \lambda = m/n$.
Let $\beta \geq 0$ be finite.
Define a Markov random process ${\cal M}$ on a finite graph $G = (V,E)$ 
whose vertex set $V$ is formed by considering the uniform mixture $\tau_0$ on the set of zero-eigenstates
of $H$, and an additional vertex for each unique state formed by applying a Pauli error applied to $\tau_0$:
$$
V := \left\{ P \cdot \tau_0 \cdot P, \quad P\in {\cal P}^n \right\}
$$
For any two vertices $\tau_i,\tau_j$ such that
$$
\tau_j = P \tau_i P
$$
where $P$ is a single qubit Pauli $P\in {\cal P}$ we define the following
transition probabilities:
$$
\forall i \neq j \quad
{\cal M}_{i,j} =
\frac{1}{4n}
\min \left\{1, \exp \left\{ \beta(E_{\tau_i} - E_{\tau_j})/\lambda \right\}\right\}
$$
and
$$
{\cal M}_{i,i} = 1 - \sum_{j\neq i} {\cal M}_{i,j}
$$
where
$$
E(\tau_i) = \tr( \tau_i H)
$$
\end{definition}

We note that under the definition above, any two vertices connected by an element of the stabilizer group $g\in {\cal G}$, i.e.
$$
\tau_i = g \tau_j g^{\dag}
$$
will correspond to the same vertex - since it preserves the uniform distribution on the codespace.
In particular we have $|V| = | \F_2^n / {\cal C}_x | = | \F_2^n / {\cal C}_z|$.
More generally for stabilizer codes, each vertex corresponds to a minimal weight error modulo the stabilizer group.

Also note, that the transition probabilities ${\cal M}_{i,j}$ correspond to a $2$-step process, where at the first step one samples
a uniformly random index $k\in [n]$ and then applies a uniformly random Pauli error ${\cal E}$ on that index
with probability corresponding to the exponent of energy differences.

We also note that the normalization by factor of $4n$ stems from the size of the single-qubit Pauli group $|{\cal P}| = 4$.

\begin{fact}
There exists a stationary distribution of ${\cal M}$, denoted by $\rho_0$
and it satisfies:
$$
\rho_0 = \frac{1}{Z} e^{-\beta \tilde H}
$$
where $\tilde H = H / \lambda$ and $Z = \tr(e^{-\beta \tilde H})$ is the partition function for value $\beta$.
\end{fact}

\begin{proof}

A Markov chain is ergodic - i.e. it has a unique stationary distribution if and only if it is simultaneously aperiodic and irreducible.
The Markov chain defined above corresponds to a finite graph.
It is aperiodic if and only if it is non-bipartite, and irreducible if and only if it is connected.
The definition above satisfies both constraints : any state can be reached from any other state (irreducibility),
and each state has a self-loop (namely - the identity error) and hence  the graph is non-bipartite.
Hence there exists a single stationary distribution $\rho_0$.

For $\tau\in V$ let $E(\tau)$ denote the energy of $H$ corresponding to $\tau$ and consider the function 
$$
\forall \tau\in V 
\quad
\pi(\tau) = e^{- \beta E(\tau)/\lambda}
$$
Let $\tau_1 \neq  \tau_2\in V$ be two states with corresponding energy values $E_1,E_2$.
Let ${\cal M}_{1,2}$, ${\cal M}_{2,1}$ denote the transition probabilities under ${\cal M}$ from $\tau_1$ to $\tau_2$,
and vice versa.
Then by definition of $\pi$ we have
\be\label{eq:detailed}
\pi(\tau_1) \cdot {\cal M}_{1,2} 
=
\pi(\tau_1) \cdot \min \left\{1, e^{- \beta(E_2 - E_1)/\lambda}\right\} = 
\pi(\tau_2) \cdot \min \left\{1, e^{- \beta(E_1 - E_2)/\lambda}\right\}
=
\pi(\tau_2) \cdot {\cal M}_{2,1}
\ee
This implies in particular that $\pi(\tau)$ 
satisfies a so-called {\it detailed balance} equation w.r.t ${\cal M}$
and so it is a stationary distribution of ${\cal M}$, up to a constant factor.
By the above, it is in fact the single stationary distribution of ${\cal M}$, and thus
$$
\rho_0 = 
\frac{1}{Z} e^{-\beta H/\lambda}, \quad Z = \tr(e^{-\beta H/\lambda})
$$
\end{proof}

\subsection{The Thermal Gibbs Markov Process for $\qLTC$'s}

As a next step, we consider a {\it truncated} random process ${\cal M}_{k}$ for integer $k$ where
one only considers errors up to some "typical" weight $k$, beyond which the measure of the stationary
distribution of the original process ${\cal M}$ is negligible.

\begin{definition}\label{def:trnc}

\textbf{$k$-Truncated Markov chain}

\noindent
Let ${\cal C}$ be a quantum stabilizer code on $n$ qubits with $m$ checks, and let $H = H({\cal C})$.
Set $\lambda(H) = \lambda = m/n$.
Let $\beta \geq 0$ be finite.
For any two vertices $\tau_i,\tau_j$ such that
$$
\tau_j = P \tau_i P
$$
where $P$ is a single qubit Pauli $P\in {\cal P}$ we define the following
transition probabilities:
\[ 
\forall i \neq j \quad
{\cal M}_{i,j} =	
\begin{cases}
0 & \mbox{if} \ \ \Delta(\tau_j,\tau_0) >k/n  \\
\frac{1}{4n}
\min \left\{1, \exp \left\{ \beta(E_{\tau_i} - E_{\tau_j})/\lambda \right\}\right\} & o/w \\
	\end{cases}
\]
where $\Delta(\tau_i,\tau_j)$ is the minimal weight of a Pauli $P$ such that $P \tau_i P = \tau_j$,
and
$$
{\cal M}_{i,i} = 1 - \sum_{i \neq j} {\cal M}_{i,j}
$$
%
\end{definition}

In general, given the energy parameter $\beta>0$ one cannot bound a so-called "typical" weight,
for which the measure of errors above that weight are negligible in the thermal Gibbs state $e^{-\beta H}$.
However, for the specific case of $\qLTC$'s such a bound is readily available, via the
soundness parameter $\eps>0$.
\begin{proposition}\label{prop:trnc}

\textbf{Truncated Metropolis Hastings Approximates the Gibbs State of a $\qLTC$}

\noindent
Suppose in particular that $H = H({\cal C})$ where ${\cal C}$ is a $(q,s)$ $\sLTC$,
and set $\lambda = \lambda(H)$.
Let $0 < \delta < 1/2$ and denote
$
k = n \delta
$.
Let $\rho_k$ denote a stationary distribution of the $k$-th truncated Markov chain ${\cal M}_k$.
Then for 
$$
\beta \geq 5 \ln(1/\delta) / s
$$
the $k$-th truncated Markov chain approximates the thermal Gibbs state of the scaled
Hamiltonian $\tilde H = H / \lambda$:
$$
\left\|  \rho_k - \frac{1}{Z} e^{-\beta \tilde H}\right\| 
\leq
2 n \cdot e^{- 2 n \cdot \ln(1/\delta) \cdot \delta }, \quad Z = \tr(e^{-\beta \tilde H})
$$
\end{proposition}

\begin{proof}

The chain ${\cal M}_{k}$ is irreducible and aperiodic.
It follows that it has a unique stationary distribution.
We denote this distribution by $\rho_k$.
In addition, every pair of neighboring vertices $\tau_i \neq \tau_j$ in ${\cal M}_k$ satisfy the detailed balance
equation \ref{eq:detailed} which is identical to the one they share in ${\cal M}$.
Hence, the marginal distribution of the complete Markov chain $\rho_0$ to the vertices of ${\cal M}_k$
is identical to $\rho_k$.

Thus, it is sufficient to place an upper bound the probability measure of $\rho_0$ on level sets
$k' > k$.
Indeed,
by the $\qLTC$ condition it follows that a minimal error ${\cal E}$ of weight
$|{\cal E}| = k$ satisfies
$$
\frac{1}{Z} e^{- \beta (|{\cal E}|/n) \cdot m / \lambda } \leq \P({\cal E}) \leq \frac{1}{Z} e^{- \beta (|{\cal E}|/n) \cdot m \cdot  s/ \lambda}
$$
and so
\be\label{eq:bound1}
\frac{1}{Z} e^{- \beta |{\cal E}| } \leq \P({\cal E}) \leq \frac{1}{Z} e^{- \beta |{\cal E}|\cdot  s}
\ee
For any integer $\ell$ the number of minimal weight errors of weight at most $\ell$ is at most
the number of Pauli operators of weight at most $\ell$.
This latter quantity can be upper-bounded by the volume of the $\ell$-th Hamming ball ${\cal B}_\ell$ as follows:
$$
2^{2\ell} \cdot {\rm Vol}({\cal B}_\ell)\leq 2^{2\ell} \cdot e^{ n \cdot H(\ell/n) }
$$
since in addition $Z \geq e^{-\beta s 0} = 1$ we have:
\be\label{eq:errorbound1}
\P( |{\cal E}| \geq k)
\leq
(n - k) \cdot \max_{\ell \geq k} \left\{ e^{ 2 \ell + n \cdot H(\ell/n)} \cdot e^{- \ell \beta s } \right\}
\ee
for $\alpha =  \ell/n \leq 1/2 $ we have
$$
H(\alpha) \leq 2 \ln(1/\alpha) \cdot \alpha 
$$
hence
$$
\P \left( |{\cal E}| \geq k \right)
\leq
\max_{\ell \geq k} n \cdot e^{2 n \cdot \alpha+ 2 n \cdot \ln(1/\alpha) \cdot \alpha } \cdot e^{- \beta s n \cdot \alpha}
$$
since $\beta s > 5 \ln(1/\delta) \geq 5 \ln(1/\alpha)$ 
we have
\be\label{eq:finalbound}
\P \left( |{\cal E}| \geq k \right)
\leq
\max_{\ell \geq k} n \cdot e^{- 2 n \cdot \ln(1/\alpha) \cdot \alpha }
\ee
Using the bound $1/(1-x) \leq 1 +  x$ it follows that:
$$
\left\|  \rho_k - \frac{1}{Z} e^{-\beta H/\lambda}\right\| 
\leq
2 \max_{\ell \geq k} n \cdot e^{- 2 n \cdot \ln(1/\alpha) \cdot \alpha }, \quad Z = \tr(e^{-\beta H/\lambda})
$$
Since RHS is maximized for $\alpha = \delta = k/n$ the proof follows.

\end{proof}

\subsection{Percolation Behavior of Random Errors in the Gibbs State of $\qLTC$'s}

\noindent
We now recall some of the definitions of Fawzi et al. \cite{FGL18}.
The first one is that of an $\alpha$-subset which is a subset that has a large intersection
with some fixed subset:
\begin{definition}

\textbf{$\alpha$-subset}

\noindent
Let $G = (V,E)$ be a graph, $X\subseteq V$, and $\alpha \in [0,1]$.
An $\alpha$-subset of $X$ is a set $S\subseteq V$ such that $|S \cap X| \geq \alpha \cdot |S|$.
We denote by ${\rm maxconn}_\alpha(X)$ as the maximum size of an $\alpha$ connected subset of $X$.
\end{definition}

The second definition is that of a locally-stochastic random error model, which generalizes an
independent random error model in that the probability of a set decays exponentially in its size:

\begin{definition}\label{def:locallystochastic}

\textbf{Locally-stochastic}

\noindent
Let $V$ be a set of $n$ elements.
A random subset $X\subseteq V$ is said to be locally-stochastic with parameter $p\in [0,1]$ if 
for every $S \subseteq V$ we have
$$
\P( X \supseteq S ) \leq p^{|S|}
$$

\end{definition}

\noindent
We now recall Theorem 17 of \cite{FGL18} on the percolation behavior of $\alpha$-subsets.
It states, roughly, that the size of the maximal $\alpha$-connected component when choosing
vertices at random with probability $p$ drops exponentially in $d p^{\alpha}$.
\begin{theorem}
(Theorem 17 of \cite{FGL18})
Let $G = (V,E)$ be a graph on $n$ vertices, such that each vertex has at most $D = D(n)$ neighboring edges.
Let
$$
p_{ls}
=
\left(
\frac{2^{-h(\alpha)}}{(D-1) (1 + 1/(D-2))^{D-2}}
\right)^{1/\alpha}
$$
where $h(\alpha)$ is the binary entropy function.
Let $X\subseteq V$ be a random subset of $V$ that is locally stochastic with parameter $p<p_{ls}$.
Then
$$
\P({\rm maxconn}_\alpha(X) \geq t) \leq 
C |V| \left(\frac{p}{p_{ls}}\right)^{\alpha t}
$$

\end{theorem}
In particular, we have:
$$
p_{ls} \geq  (2De)^{-1/\alpha}
$$
and so
$$
\left( \frac{p}{p_{ls}}
\right)^{\alpha t}
=
(2D e p^{\alpha})^{t}
$$
using this,
we rephrase the theorem as the following lemma:
\begin{lemma}\label{lem:percolation}

\textbf{Percolation behavior for locally-stochastic random errors}

\noindent
Let $G = (V,E)$ be a graph on $n$ vertices, such that each vertex has at most $D = D(n)$ neighboring edges.
Let $\alpha>0$.
Let $X\subseteq V$ be a random subset of $V$ that is locally stochastic with parameter $p$.
There exists a constant $c$ such that if $p < c / D$ we have
$$
\P({\rm maxconn}_\alpha(X) \geq t) \leq 2 n \cdot (2 D e p^{\alpha})^t
$$
\end{lemma}

Consider now a local Hamiltonian $H$, we define its interaction graph as follows:
\begin{definition}\label{def:interact}

\textbf{Interaction graph of a local Hamiltonian}

\noindent
Let $H = \sum_i H_i$ denote a local Hamiltonian on $n$ qubits.
The interaction graph of $H$, $G(H) = (V,E)$ is defined by $V = [n]$
corresponding to the $n$ qubits, and $e = (i,j)\in E$ if qubits $i$ and $j$
share a local term $H_e$ in $H$.
\end{definition}
We would like to show that the $k$-th truncated Metropolis-Hastings random process on $H$
is locally-stochastic for sufficiently small $k$.

To see why this is a non-trivial statement, 
recall that the MH random process does not induce independent errors, since the probability
of adding error to a given qubit depends on the additional energy cost induced by
flipping that qubit, and that additional energy depends on the specific error configuration
on its neighboring qubits.

In fact this random error model implies that errors are {\it more likely} to occur near previously sampled
errors thus leading to a behavior that is completely opposite to local stochasticity.
However, we show that if $k$ is significantly less than $n/D$ then 
this effect is negligible compared to the probability of sampling an error that is not connected
to any other error, and hence approximately
these errors are
locally-stochastic.

\begin{lemma}\label{lem:localstoc}

\textbf{The Thermal Gibbs State is Locally-Stochastic}

\noindent
Let ${\cal C}$ be a stabilizer code and let $H = H({\cal C})$ denote the corresponding
local Hamiltonian.  Suppose that the corresponding interaction graph $G(H)$ has degree at most $D$.
Let $\alpha \in (0,1]$, and
consider the $k$-th truncated Markov chain ${\cal M}_k$ and its stationary distribution $\rho_k$, for 
$$
k \leq \frac{n}{2e(D e^{300})^{1/\alpha}} 
$$
If the energy density is sufficiently large compared to the inverse temperature:
$$
\lambda \geq \beta \ln(n)
$$
then ${\cal E} \sim \rho_k$ is locally-stochastic 
with parameter at most
$$
p_0 \leq 2 k e / n
$$
with probability at least $1 - (k+1) n^{-4}$.

\end{lemma}

\begin{proof}
For all $i <  k$ let ${\cal E}_i$ denote a random error sampled according to the marginal distribution of $\rho_k$
to errors of minimal weight $i$.
For an error ${\cal E} \in {\cal P}^n$ let $E({\cal E})$ denote the energy of ${\cal E}$ w.r.t. $H$:
$$
E({\cal E}) = \tr(H \cdot {\cal E} \cdot \rho_{gs} \cdot {\cal E})
$$
where $\rho_{gs} \in \ker(H)$.
By the definition of ${\cal M}_k$, 
clustering its vertices according to their minimal weight results in a Markov chain whose graph is a layered graph. 
In particular,
${\cal E}_{i+1}$ can be simulated 
as the following rejection sampler:
\begin{enumerate}
\item  
Sample an error according to ${\cal E}_i$.
\item\label{it:2}
Sample a uniformly random qubit $\ell \sim U[n]$, then choose $P_\ell = X$ w.p. $1/2$ and $P_\ell = Z$ w.p. $1/2$.
Set ${\cal E} = I \otimes I \otimes \hdots \otimes P_{\ell} \otimes I \otimes \hdots \otimes I$
w.p.
$$
\min \left\{ 1, e^{ (\beta / \lambda) \cdot \left( E( {\cal E}_i \cdot {\cal E} ) - E({\cal E}) \right) } \right\}
$$
\item\label{it:3}
Update ${\cal E}_{i+1} = {\cal E}_i \cdot {\cal E}$ only if it increases the error weight modulo ${\cal C}$, i.e. if:
$$
|{\cal E}_{i+1}| = i+1
$$
\end{enumerate}

\noindent
We place the following induction hypothesis for $i \leq k-1$: 
$$
\P \left({\cal E}_i \mbox{ is locally stochastic with parameter } p \leq 2ie / n \right) \geq 1 - i n^{-4}
$$
The base case $i=1$ is trivial since ${\cal E}_1$ has only single qubit errors.
We assume the hypothesis for arbitrary $i$
and show that ${\cal E}_{i+1}$ is locally stochastic with parameter at most $2(i+1)e/n$ w.h.p.

Consider a subset $S \subseteq [n]$, fix some $j\in S$ and let $S_{-j}$ denote $S$ with $j$ removed.
The probability that $S$ is contained in ${\cal E}_{i+1}$ is equal by definition to:
$$
\P( S \subseteq {\cal E}_{i+1} ) = 
\P( S \subseteq {\cal E}_i ) + \sum_{j\in S} \P( S_{-j} \subseteq {\cal E}_i \wedge j \notin {\cal E}_i \wedge j\in {\cal E})
$$
\be\label{eq:ei}
=
\P( S \subseteq {\cal E}_i ) + \sum_{j\in S} \P( S_{-j} \subseteq {\cal E}_i \wedge j \notin {\cal E}_i )
\cdot 
\P( j\in {\cal E} | S_{-j} \subseteq {\cal E}_i \wedge j \notin {\cal E}_i)
\ee
At this point we require an upper-bound on the rightmost multiplicative term above. 
We first claim that ${\cal E}_i$ itself has only very sparse errors:
\begin{proposition}\label{prop:aux1}
W.p. at least $1 - (i+1) \cdot n^{-4}$ the largest connected component of the random error ${\cal E}_i$ has size at most $\ln(n)/100$.
\end{proposition}

\begin{proof}
By induction assumption we have that w.p. at least $1 - i n^{-4}$ the random error 
${\cal E}_i$ is locally-stochastic with parameter 
$$
p_0 \leq 2i e/n \leq 2ke/n \leq (D e^{300})^{-1/\alpha} 
$$
so
$$
D \cdot p_0^{\alpha} = D \cdot  (D e^{300})^{-1} = e^{-300}
$$
applying percolation lemma \ref{lem:percolation} we have that in such a case
$$
\P \left( {\rm maxconn}_\alpha{\cal E}_i > \ln(n)/100\right) 
\leq
2 n \cdot \left( 2 D e (p_0)^{\alpha}  \right)^{2\ln(n)/100}
$$
$$
\leq
2n \cdot e^{-5.5\ln(n)} \leq n^{-4}
$$
It follows by the union bound that w.p. at least 
$$
1 - i \cdot n^{-4} - n^{-4} = 1 - (i+1) n^{-4}
$$
the random error ${\cal E}_i$'s largest connected component is of size at most $\ln(n)/100$.
\end{proof}

We now leverage this proposition to argue that the bias in favor of any specific qubit to be chosen
as a new error is very moderate:
\begin{proposition}\label{prop:aux2}
Let $i<k$ and suppose that
${\cal E}_i$ has no connected components of size exceeding $\ln(n)/100$.
Then
$$
\P( j\in {\cal E} | 
S_{-j} \subseteq  {\cal E}_i \wedge j \notin {\cal E}_i
)  \leq \frac{2e}{n}
$$
\end{proposition}

It is perhaps insightful to consider at this point why the claim above is non-trivial: once some error of weight $i$ is fixed
that contains all but a single qubit $j$ of $S$, it is not clear why the probability that $j$ is selected
upon transition to ${\cal E}_{i+1}$ is negligible: it could be the case that because $j$ has many neighboring errors in $S \cap {\cal E}_i$
the additional energy the system is penalized for when adding error on $j$ is actually {\it less} than on other qubits - 
perhaps so much less that it is {\it more} likely to select $j$ than any other qubit.
We show however that this is not the case because w.h.p. ${\cal E}_i$ has very few large clusters
so all qubits have comparable probability of being selected:
\begin{proof}
%
At each step $i$, the number of qubits $\ell$ satisfying item \ref{it:3}, i.e. those for which $|{\cal E} \cdot {\cal E}_i| = i+1$ is $n-i$.
On the other hand, since the maximal connected component is of size at most $\ln(n)/100$ 
then each qubit shares at most $\ln(n)/100$ checks with other qubits.
It follows that the minimal number of checks that are violated by adding a new error to ${\cal E}_i$ is at least
$$
D - \ln(n)/100
$$
and so the relative probability $p$ of adding a new error at step $i+1$ satisfies by step \ref{it:2}
$$
e^{- \beta (D - \ln(n) /100) / \lambda} \geq p \geq e^{-\beta D / \lambda}
$$
It follows that in such a case the random error ${\cal E}$ 
has the property that the ratio of error probability between any pair of qubits satisfying \ref{it:3} is at most
\be\label{eq:effect}
e^{\beta \ln(n)/100\lambda} \leq e
\ee
where we have used the assumption that $\lambda \geq \beta \ln(n)$.
Hence the probability that some fixed $j$ is chosen by ${\cal E}$ at step $i+1$ is at most 
$$
\P( j\in {\cal E} | 
S_{-j} \subseteq  {\cal E}_i \wedge j \notin {\cal E}_i
) \leq \frac{e}{n - i} \leq \frac{2e}{n}
$$
where we have used the fact that $i \leq k \leq  n / 4 D \leq n/2$.
\end{proof}

\paragraph{Completion of proof:}

Applying propositions \ref{prop:aux1} and \ref{prop:aux2} to Equation \ref{eq:ei} we conclude that w.p. at least $1 - (i+1) n^{-4}$
the probability that $S\in {\cal E}_{i+1}$
is upper-bounded by the following expression:
$$
\P( S \subseteq {\cal E}_{i+1} )
\leq
\P( S \subseteq {\cal E}_i ) + \sum_{j\in S} \P( S_{-j} \subseteq {\cal E}_i \wedge j \notin {\cal E}_i )\cdot (2e/n) 
$$
Assume that this is the case.  Then, in addition, the induction hypothesis on ${\cal E}_i$ holds and we have:
$$
\forall T \quad \P( T \subseteq {\cal E}_i ) \leq  (2ie/n)^{|T|}
$$
hence
$$
\P( S \subseteq {\cal E}_{i+1} )
\leq
(2ie/n)^{|S|} + |S| \cdot (2ie/n)^{|S|-1} \cdot (2e/n)
$$
$$
\leq
(2ie/n)^{|S|} \cdot (1 + |S| / i)
$$
we can place an upper bound on the right factor using the binomial:
$$
(1 + |S| / i) \leq \sum_{m \leq |S|} {|S| \choose m} i^{-m} = (1 + 1/i)^{|S|} = \left( \frac{i+1}{i} \right)^{|S|}
$$
hence w.p. at least $1 - (i+1) n^{-4}$ we have:
$$
\P( S \subseteq {\cal E}_{i+1} )
\leq
(2ie/n)^{|S|} \cdot \left( \frac{i+1}{i} \right)^{|S|}
=
\left(\frac{2e(i+1)}{n}\right)^{|S|}
$$
This concludes the proof by induction.

Finally, since for each $i$, the random process ${\cal E}_i$ is locally stochastic with parameter at most $2 i e/n$
with probability at least $1 - (i+1) \cdot n^{-4}$ then
any convex mixture of ${\cal E}_i$ - and in particular, the stationary distribution $\rho_k$ is also
locally-stochastic with parameter at most $2ke/n$ w.p. at least $1 - (k+1) n^{-4}$.
\end{proof}

We conclude our central lemma of this section - which is
that the thermal Gibbs state $e^{-\beta H}$ where $H$ is a Hamiltonian corresponding to a $\qLTC$,
and $\beta$ is sufficiently large, 
satisfies a percolation property - namely that the maximal $\alpha$-connected component
of a typical error ${\cal E}$ is of logarithmic size:
\begin{lemma}\label{lem:maxconn}

\textbf{Typical error components are small for the thermal state of $\qLTC$'s}

\noindent
Let ${\cal C}$ be a $(q,s)$-$\sLTC$ on $n$ qubits 
and let $H({\cal C})$ be its corresponding Hamiltonian, $\lambda(H) = \lambda$.
Suppose that the interaction graph of $H$,  $G(H)$, is of degree most $D$.
Let $\alpha>0$.
Let 
$$
\tau = {\cal E} \cdot \tau_0 \cdot {\cal E}
$$ 
be a random state 
(the uniform code mixed state conjugated by a random error ${\cal E}$)
sampled according to the distribution $e^{-\beta H/\lambda}/Z$
for
$$
(10/\alpha) \cdot \ln(D)/ s
\leq
\beta
\leq
\lambda/ \ln(n)
$$
Then 
$$
\P \left( {\rm maxconn}_\alpha{\cal E} > \ln(n)/100 \right) \leq n^{-3}
$$
\end{lemma}

\paragraph{The range of values $\beta$:}
It is insightful at this point to consider the statement of the lemma w.r.t. the parameter $\beta$:
 the statement of the lemma requires that $\beta$ is within
some range - between $\log\log(n)$ and $\log(n)$.
This initially might seem strange as intuitively, increasing $\beta$
can only improve the ability to correct errors since it corresponds 
to a regime of much fewer errors - e.g. lower temperature.

However in Lemma \ref{lem:localstoc} it turns out that
the analysis is more subtle:
 indeed
we require $\beta$ to be also sufficiently small so that the error model is locally stochastic:
if $\beta$ is too large (i.e. the temperature is very low) it turns
out that a qubit that is hit by an error is much more likely to be
hit by another error - this contrary to local stochasticity,
whereas
for higher temperatures this phenomenon is greatly alleviated.
The quantitative analysis of this effect is specifically captured
in Equation \ref{eq:effect}.

Hence, the phenomenon of locally stochastic errors, that we exploit
to demonstrate a shallow decoder is in fact relevant only for a median
range of temperatures: for very low temperatures, the error is no longer
locally stochastic, but in that range - the {\it absolute} number of errors
is extremely small to allow worst-case error correction.
For higher temperatures, the absolute number of errors is very large
but conforms to the locally stochastic model which is treatable by a local
decoder.
This results in a "win-win" situation, which is handled case-by-case
in the proof of the main theorem.

\begin{proof}

Consider the $k$-th truncated MH process ${\cal M}_k$
for 
$$
k = \frac{n}{2(D e^{300})^{1/\alpha}}  
$$
By Lemma \ref{lem:localstoc} for values of $\beta, \lambda$ specified in the assumption we have:
\be\label{eq:rhok}
\P \left( \rho_k \mbox{ is locally stochastic with parameter } p_k \leq (D e^{300})^{-1/\alpha} \right) \geq
1 - (k+1) n^{-4}
\ee
We have
$$
D \cdot p_k^{\alpha} = D \cdot  (D e^{300})^{-1} \leq e^{-300}
$$
so applying percolation lemma \ref{lem:percolation} 
and the union bound w.r.t. Equation \ref{eq:rhok} we have that typical errors in the
stationary distribution of the $k$-th truncated MH process have very small
$\alpha$-connected components:
\begin{align}\label{eq:rhok3}
\P_{\rho_k} \left( {\rm maxconn}_\alpha{\cal E} > \ln(n)/100\right)  
&\leq
2 n \cdot \left( D (p_k)^{\alpha}  \right)^{2\ln(n)/100} + (k+1) n^{-4} \nonumber \\
&\leq
2n \cdot e^{-6\ln(n)} + (k+1) n^{-4} \leq n^{-3} /2
\end{align}
On the other hand, by Proposition \ref{prop:trnc} the stationary distribution $\rho_k$ of ${\cal M}_k$ approximates
the Gibbs state of the $\sLTC$ - namely
$\rho_0 = e^{- \beta H / \lambda}$ for:
$$
\beta  \geq 5 \ln(n/k)/s
\equiv
5 \ln(1/\delta)/s
$$
up to error at most
\be\label{eq:rhok2}
\left| \rho_k - \rho_0 \right| \leq 2 n e^{-2 n \delta \ln(1/\delta)} \leq n^{-5}
\ee
where we've substituted $\delta = k/n$.
Taking the union bound w.r.t. Equations \ref{eq:rhok3} and \ref{eq:rhok2}
and substituting our choice for $k$ we have that for
$$
\beta 
\geq
(10/\alpha) \cdot \ln(D)/s
\geq 
5 \ln(1/\delta)/ s
$$
we have
$$
\P_{\rho_0} \left( {\rm maxconn}_\alpha{\cal E} > \ln(n)/100\right) \leq n^{-3}
$$

\end{proof}


\section{A Shallow Decoder for Low Error Rate}

The last component of the proof 
is to demonstrate a shallow circuit that can correct the thermal state $e^{-\beta H}$
to a code-state, for sufficiently large $\beta>0$ (finite or not).
In the previous section we've seen that such a state can be modeled as a random
error process with small rate.
We would now like to leverage that understanding, together with the small-set expansion
property of the $n$-projective cube to show that the quantum version of the Sipser Spielman
decoder yields a shallow decoder.

Inspired by the decoding algorithm of Fawzi et al. we propose an algorithm for decoding 
a random error ${\cal E}$ in depth proportional to $\log({\rm maxconn}_\alpha({\cal E}))$.
It is based on a parallel version of the Sipser-Spielman decoder:

%
%
%

%
%
%
%

We first rephrase the original parallel Sipser-Spielman decoder as an algorithm that can decode
errors on a binary code of $\F_2^n$ with a slightly relaxed condition.
Instead of requiring the bi-partite graph of the code to be expanding, we merely
ask that the set of errors expands significantly 
in the Tanner graph of the code
at each step:

\begin{lemma}\cite{SS96}(Theorem 11)\label{lem:ss}
\textbf{Parallel decoder for small-set expander graphs}

\noindent
Let $C$ be a code on $n$ bits and let $G$ denote the Tanner graph of $C$.
Suppose $G$ is a $(c,d)$-bi-regular graph on $n$ vertices.
The parallel decoder ${\cal A}$
is an algorithm that given error ${\cal E} = {\cal E}_1$ iteratively replaces it with errors ${\cal E}_i$ for $i \geq 1$.
At step $i$ the algorithm may modify bits only in the support of ${\cal E}_i \cup \Gamma({\cal E}_i)$,
and in particular, examines for each bit $k$ only $\Gamma(k)$.
If, in addition, at the beginning of iteration $i$ we have:
$$
|\Gamma({\cal E}_i)| \geq |{\cal E}_i| \cdot c \cdot (3/4 + \eps)
$$
for some constant $\eps>0$,
then after step $i$ the weight of the residual error ${\cal E}_i$ decreases by a multiplicative factor:
$$
|{\cal E}_{i+1}| \leq |{\cal E}_i| \cdot (1 - 4 \eps)
$$
\end{lemma}

Our quantum decoder is an application of the Sipser-Spielman decoder
on the individual $X,Z$ errors.

\begin{algorithm}\label{alg:main}

\textbf{Shallow Decoder ${\cal B}$}

\noindent
\item
Input: a quantum state $\rho$ on $n$ qubits, a set of $X$ checks ${\cal C}_x$ and
a set of $Z$ checks ${\cal C}_z$.
\begin{enumerate}
\item
Run the decoder ${\cal A}$ w.r.t. $Z$ errors using ${\cal C}_x$.
\item
Run the decoder ${\cal A}$ w.r.t $X$ errors using ${\cal C}_z$.
\end{enumerate}
\end{algorithm}

\begin{lemma}\label{lem:decode}
Consider the projective code ${\cal C} = ({\cal C}_x, {\cal C}_z)$ on $n$ qubits with $p = n/2$, and let ${\cal E} \in {\cal P}^n$ denote an error 
with far-away and small connected components:
$$
{\rm maxconn}_{\alpha}({\cal E}) \leq \ln(n)/100, \quad \alpha = 1/ (\gamma \log\log(n))
$$
where $\gamma = \log(1 - 4 \cdot (3/16))$ is the constant implied by Lemma \ref{lem:ss} for $\eps = 3/16$.
Then shallow decoder ${\cal B}$ runs in depth at most $2\gamma \log^2 \log(n)$ steps and satisfies:
$$
{\cal B} \circ {\cal E} \circ \rho = \rho \quad \forall \rho\in {\cal C}
$$
\end{lemma}

\begin{proof}
Let ${\cal E}_i$ denote the set of erred qubits at step ${\cal E}_i$, with ${\cal E}_1 = {\cal E}$ denoting the initial error.
By the first property of Lemma \ref{lem:ss} at each step $i\in [t]$ error ${\cal E}_i$ is supported on
qubits at distance at most $t$ from the initial error ${\cal E}$:
$$
{\rm supp}({\cal E}_i) \subseteq \Delta_t({\cal E})
$$
where $\Delta_t({\cal E})$ is the set of qubits at distance at most $t$ from ${\cal E}$ in the Tanner
graph of ${\cal C}_x$ or in the Tanner graph of ${\cal C}_z$.
In addition, by the monotonicity of error weight in Lemma \ref{lem:ss} we have
$$
\forall i>1 \quad |{\cal E}_i| \leq |{\cal E}_{i-1}|
$$
It follows that the union
$$
\hat{\cal E}_i := \bigcup_{j=1}^i {\cal E}_j
$$
is an $\alpha$-subset of ${\cal E}$
with $\alpha \geq 1/t$.
Since by assumption 
$$
{\rm maxconn}_{\alpha}({\cal E}) \leq \ln(n)/100
$$
Hence if the number of decoding iterations is sufficiently small, i.e. $t \leq 1/\alpha$ then we have:
$$
\forall i\in [t] \quad
{\rm maxconn}({\cal E}_i) \leq 
{\rm maxconn}(\hat{\cal E}_i) \leq
{\rm maxconn}_{\alpha}({\cal E}) \leq
\ln(n)/100
$$ 
Lemma \ref{lem:projexp} then implies that for each $i\in [t]$ error ${\cal E}_i$ has large expansion as follows:
$$
|\partial {\cal E}_i| \geq |{\cal E}_i| \cdot D_x \cdot (15/16)
\quad
|\partial {\cal E}_i| \geq |{\cal E}_i| \cdot D_z \cdot (15/16)
$$
where $D_x,D_z$ are, respectively, the degree of each qubit in checks ${\cal C}_x, {\cal C}_z$.

Thus, so long as $t \leq 1/\alpha$ the error pattern ${\cal E}_i$ for each $i\in [t]$ 
expands with factor at least $(3/4 + 3/16)$ so Lemma \ref{lem:ss} 
is applicable at each step $i$ with $\eps = 3/16$.
It follows that for each of these $X, Z$ error types the decoder algorithm ${\cal A}$ runs in a number of iterations 
which is at most the number of steps to decode the maximal-size connected component:
$$
t \leq \gamma \log({\rm maxconn}({\cal E})) \leq  \gamma \log\log(n)
$$
Thus, in time at most $t$ the decoder corrects all errors, provided $t \leq 1/\alpha$.
Since we chose $\alpha = 1/ \gamma \log\log(n)$ then the assumption $t \leq 1/\alpha$ is correct.

Finally, analyzing the depth of the quantum circuit implementing the decoder:
Since at each step only the neighboring checks on any given bit are examined, and each qubit
is incident on $O(\log(n))$ checks
it follows that shallow decoder ${\cal B}$
corrects ${\cal E}\circ \rho$ and runs in depth at most $2\gamma (\log\log(n))^2$.
\end{proof}
We note here that the decoder ${\cal B}$ requires extra ancillary bits for syndrome computation,
hence the notation ${\cal B} \circ {\cal E} \circ \rho$ signifies a quantum channel, where some
of the qubits are discarded after computation.

\section{Global Entanglement for Thermal States}

\subsection{The construction}

\begin{enumerate}
\item
\textbf{Step 1 - The projective code:} 

\noindent
Fix $n$ as the number of qubits in the code.
As the basis for our construction we consider the $(N,p)$ projective code ${\cal C}$ for $p = N/2$.
By Lemma \ref{lem:LLZ18} we can choose $N = \Theta(\log(n))$ such that ${\cal C}$ is a $\qLTC$ $[[n,1,n^c]]$ for some $c>0$ 
with $\qLTC$ parameters 
$$
(q = \log(n), s = 1/\log^2(n)).
$$ 
By construction, the interaction graph of $H({\cal C})$, i.e. $G(H({\cal C}))$ 
is $D$-regular with $D = 2\cdot \log(n)$.
The local Hamiltonian $H$ has $m = 2 n \cdot \log(n)$ check terms.
\item
\textbf{Step 2 - Amplification:}

\noindent
We apply 
Proposition \ref{prop:avg}
to conclude the existence of a $\qLTC$, denoted by ${\cal C}'$
with parameters
$$
(q' = \ceil{\log^3(n)}, s' = 1/e)
$$
and $\lambda = \log^4(n)$.
The interaction graph of the Hamiltonian of ${\cal C}'$, i.e. $G(H({\cal C}'))$ has
degree at most $D' \leq \ceil{\log^7(n)}$.
\item
\textbf{Step 3 - Union:}

\noindent
Finally, we consider the union of the checks of ${\cal C}$ and ${\cal C}'$
and denote the union by ${\cal C}_{pa}$ - this is our construction.
We denote the number of checks by $m_{pa}$.
We have that, ${\cal C}_{pa}$ is $[[n,1,n^c]]$ quantum code, and is 
$\qLTC$ with parameters:
$$
(q_{pa} = \log^3(n), s_{pa} = 1/2e, D_{pa} \leq 2\log^7(n))
$$
and $\lambda_{pa} \geq 2 \log^4(n)$.
\end{enumerate}

We note that the the amplified code ${\cal C}'$
has constant soundness for all {\it non-zero} distances, but it is not clear a-priori
why it should also satisfy $\ker({\cal C}) = \ker({\cal C}')$.
Hence, the union of ${\cal C}$ and ${\cal C}'$ is taken in order to enforce the ground-state
of the final code to equal that of ${\cal C}_{pa}$.
This slightly reduces the soundness, and increases the degree of the interaction graph of
the final code.
Also note that the check terms of ${\cal C}_{pa}$ commute in pairs.

\subsection{Main Theorem}

We now state formally our main theorem and prove it: 
\begin{theorem}\label{thm:main}
Let ${\cal C}_{pa}$ denote the code constructed above on $n$ qubits, and 
let $H = H({\cal C}_{pa})$, $\lambda = \lambda(H)$ and inverse temperature:
$$
\beta \geq 20 e \cdot \log^2\log(n)
$$
Any quantum circuit $U$ on $a \geq n$ qubits
that
approximates the thermal state of $\tilde H = H / \lambda$ on a set of qubits $S$, $|S| = n$, at inverse temperature $\beta$,
$$
\left\| \tr_{-S}(U \ketbra{0^{\otimes a}}U^{\dag}) - e^{-\beta \tilde H}/Z \right\|_1 \leq 0.1 n^{-2}
$$
has depth at least
$$
d(U) = \Omega(\ln(n))
$$ 
%
\end{theorem}

\begin{proof}

By construction, the code ${\cal C}_{pa}$ is comprised of a set of pairwise commuting projections.
We now analyze the circuit lower bound for $T=0$ and $T>0$ separately.
\item

\textbf{Case of $T=0$: (ground-state)}

\noindent
First, we consider the case of $T=0$ - i.e. $\beta \to \infty$:
Since the check terms of $H$ include those of the original code $H({\cal C})$
and the rest of the checks correspond to subsets of the checks of $H$
then 
$\ker(H) = \ker(H({\cal C}))$ so
by Lemma \ref{lem:depth} it follows that if 
\be\label{eq:bound2}
\left\| \tr_{T - S}(U \ketbra{0^{\otimes a}}U^{\dag}) - \rho_{gs} \right\|_1 \leq n^{-2}
\ee
for  $\rho_{gs}\in \ker(H)$ then
$$
d(U) = \Omega(\ln(n)).
$$ 
By definition, the Gibbs state for $\beta \to \infty$ is such a code-state and this implies the proof for $T=0$.

\item

\textbf{Case of $0< T < 2/\log^3(n)$: (error is not locally-stochastic, but has small weight)}

\noindent
Now we consider the case of finite $\beta\geq 0$ and specifically $\beta = \Omega(\log(n)^3)$.
We use Equation \ref{eq:errorbound1} by which the probability of an error of weight at least $k$ is
at most:
$$
\P( |{\cal E}| \geq k)
\leq
(n - k) \cdot \max_{\ell \geq k} \left\{ e^{ 2 \ell + n \cdot H(\ell/n)} \cdot e^{- \ell \beta s } \right\}
$$
so for $\beta s = \Omega(\log^3(n))$ the above is maximized for $k = 1$
hence
$$
\P( {\cal E} \neq 0 ) \leq n \cdot e^{-\log^3(n)} \leq n^{-4}.
$$
Hence it follows that such a state satisfies Equation \ref{eq:bound2}
which in turn implies by Lemma \ref{lem:depth} the lower bound $d(U) = \Omega(\ln(n))$.

\item

\textbf{Case of $T \geq 2/\log^3(n)$ (error has large weight, but is locally-stochastic)}

\noindent
Let
$$
\tilde \rho := \tr_{T - S}(U \ketbra{0^{\otimes n}}U^{\dag}),
\quad
\| \rho_0 - \tilde \rho \| \leq 0.1 n^{-2}
$$
where $\rho_0 = e^{-\beta \tilde H}$.
Consider a realization of the thermal state $\rho_0$ as convex mixture of ${\cal E} \cdot \rho_{gs} \cdot {\cal E}$
where $\rho_{gs}\in {\cal C}$ and ${\cal E}$ is a minimal weight error.
Let $\alpha = \Theta(1/\log\log(n))$ be the number implied by Lemma \ref{lem:decode}.
By  Lemma \ref{lem:maxconn} and the triangle inequality we have
that the typical $\alpha$-connected component sampled from
the approximate thermal state $\tilde \rho$ is of logarithmic size:
\be\label{eq:rho0}
\P_{\tilde \rho}( {\rm maxconn}_\alpha({\cal E}) \geq \ln(n)/100 ) \leq n^{-3} + 0.1 n^{-2} \leq n^{-2}
\ee
for any
$$
\beta  \geq (10/\alpha) \cdot \ln(D_{pa})  / s_{pa}
= O( \log^2 \log(n)),
\quad
\lambda \geq \beta \ln(n)
$$
where in the last inequality above we have used 
$\alpha = \Theta(1/ \log\log(n))$ and the parameters of ${\cal C}_{pa}$
by construction:
$$
D_{pa} \leq 2\log^5(n), 
s_{pa} = 1/2e,
\lambda_{pa} \geq 2 \log^4(n) \geq \beta \ln(n)
$$
Assume the error of a sampled state $\tau$ has a small maximal connected component, i.e.:
$$
\tau = {\cal E} \cdot \rho_{gs} \cdot {\cal E}, \quad {\rm maxconn}_{\alpha}{\cal E} \leq \ln(n)/100, \quad \rho_{gs} \in {\cal C}
$$
By Lemma \ref{lem:decode} it follows that using the original checks of the code ${\cal C}$ there exists
a quantum circuit ${\cal B}$ (using extra ancilla bits for syndrome computation),
\be\label{eq:db}
d({\cal B}) = O(\ln^2\ln(n))
\ee
such that
$$
{\cal B} \circ \tau = \rho_{gs},  \quad \rho_{gs}\in {\cal C}
$$
hence, by 
Equation \ref{eq:rho0}
the approximate thermal state $\tilde \rho$
generated by $U$ can be decoded into an approximate code-state
\be\label{eq:correct}
\left\|{\cal B} \circ \tilde \rho - \rho_{gs} \right\|_1 
\leq 
n^{-2}
\ee
It follows there exists a quantum circuit $V$ (possibly using extra ancilla bits)
of depth at most
$$
d(V) \leq d(U) + d({\cal B})
$$
such that
$$
\left\| V \circ \ketbra{0^{\otimes n}} - \rho_{gs}\right\|_1 \leq n^{-2}
$$
Hence
by Lemma \ref{lem:depth}
we have
$$
d(V) = \Omega(\ln(n))
$$
together with Equation \ref{eq:db} we have:
$$
d(U) = \Omega(\ln(n)).
$$

\end{proof}

\section*{Acknowledgements}
The author thanks Dorit Aharonov, Simon Apers,  Aram Harrow, Matthew Hastings and Anthony Leverrier for their useful comments and suggestions. 
He also thanks anonymous reviewers for their helpful comments and suggestions.

\bibliographystyle{alphaurl}
\bibliography{refs}

\end{document}